\documentclass[journal,letterpaper]{IEEEtran}

\usepackage{amsfonts}       
\usepackage{booktabs}
\usepackage{amsmath}
\usepackage{amsthm}
\usepackage{algorithm}
\usepackage{algpseudocode}
\usepackage{comment}
\usepackage{bbm}
\usepackage{graphicx}
\usepackage{float}
\usepackage{color}
\usepackage{hyperref}
\usepackage{subfigure}
\usepackage{cite}
\usepackage{enumitem}
\usepackage{dsfont}
\allowdisplaybreaks
\newtheorem{corollary}{Corollary}

\newtheorem{theorem}{Theorem}
\newtheorem{assumption}{Assumption}
\newtheorem{lemma}{Lemma}

\algnewcommand{\LeftComment}[1]{\(\triangleright\) #1}

\begin{document}

\title{Multi-player Multi-armed Bandits with Collision-Dependent Reward Distributions}

\author{Chengshuai~Shi and Cong~Shen
\thanks{The work was supported in part by the US National Science Foundation (NSF) under Grant CNS-2002902 and ECCS-2029978, and a Virginia Commonwealth Cyber Initiative (CCI) cybersecurity research collaboration grant.}
\thanks{The authors are with the Charles L. Brown Department of Electrical and Computer Engineering, University of Virginia, Charlottesville, VA 22904, USA. E-mail: \texttt{\{cs7ync,cong\}@virginia.edu}.}
}

\maketitle

\begin{abstract}

We study a new stochastic multi-player multi-armed bandits (MP-MAB) problem, where the reward distribution changes if a collision occurs on the arm. Existing literature always assumes a zero reward for involved players if collision happens, but for applications such as cognitive radio, the more realistic scenario is that collision \emph{reduces} the mean reward but not necessarily to zero. We {focus on the more practical no-sensing setting where players do not perceive collisions directly}, and propose the Error-Correction Collision Communication (EC3) algorithm that models implicit communication as a reliable communication over noisy channel problem, for which random coding error exponent is used to establish the optimal regret that no communication protocol can beat. Finally, optimizing the tradeoff between code length and decoding error rate leads to a regret that approaches the centralized MP-MAB regret, which represents a natural lower bound. Experiments with practical error-correction codes on both synthetic and real-world datasets demonstrate the superiority of EC3. In particular, the results show that the choice of coding schemes has a profound impact on the regret performance.

\end{abstract}


\section{Introduction}
\label{sec:intro}

The multi-armed bandits (MAB) problem is a simple yet powerful model of  sequential decision experiments with an exploration-exploitation tradeoff \cite{Bubeck:2012,Wang2018tsp,Shen2019}. In addition to its practical utility, research on MAB has contributed several impactful principles (e.g., Thompson sampling \cite{Thompson1933}, Gittins index \cite{Gittins1979}, and Optimism in the Face of Uncertainty \cite{Auer:2002,Bubeck:2012}) that prove to be useful in other related fields as well.

The multi-player version of MAB problems, in which multiple players simultaneously play the bandit game in a \textit{fully decentralized} fashion, has sparked significant interest recently \cite{gai2014distributed,tekin2012online,rosenski2016multi,besson2017multi,boursier2018sic,boursier2019practical}. Almost all of the existing works on multi-player multi-armed bandits (MP-MAB) assume that {collisions eliminate the reward on the arm}, i.e., if more than one player select the same arm simultaneously, all involved players receive zero reward. This assumption has been largely based on the application of Cognitive Radio (CR) \cite{liu2010distributed,anandkumar2011distributed,avner2014concurrent}, where it is argued that if multiple secondary users simultaneously attempt to access a channel, the mutual interference causes all communications to fail, leading to a zero reward. In practical CR systems, this is indeed the case for the carrier-sense multiple access with collision avoidance (CSMA/CA) protocol, which is widely used in WiFi \cite{CSMA2012}.

As the communication technologies advance, however, this collisions-lead-to-failed-communication principle may not always hold in modern wireless systems, in which {adaptive coding and modulation} allows reduced-rate communications to be successful even in the presence of multi-user interference \cite{SesiaLTE}. In other words, for some advanced communication protocols, a better model is to assume that \emph{collisions only impair rewards} (i.e., reduced-rate but still successful communication) instead of leading to no rewards (i.e., completely failed communication). This new communication paradigm for CR thus motivates the study of a different MP-MAB model in which collisions lead to impaired average rewards, but not necessarily zero \cite{bande2019multi}. 

As it turns out, this seemingly simple change of the model adds significant technical challenges to solving the new MP-MAB problem. In the previous model, collisions unequivocally generate  \textit{constant} rewards (i.e., zero) for all involved players, which provides unmistakable, non-random information. As a result, the only uncertainty comes from the original bandit model (the reward of each arm follows an unknown distribution and the player may only receive samples of this distribution) when there is no collision. In this new model, however, the boundary between collision and non-collision is very murky: collisions only reduce the \textit{mean} reward but the actual reward is still a random variable generated from this reduced-mean distribution, which adds significant uncertainty to the decision maker.

Another limitation of the prior research is that it mostly considers the \textit{collision-sensing} model, in which collision events (two or more players simultaneously pull the same arm) are perfectly known to the involved players \cite{rosenski2016multi,liu2010distributed,besson2017multi,boursier2018sic}. Nevertheless, in practice collisions are difficult to detect, and the more realistic scenario corresponds to the \textit{no-sensing} model in which players can only observe the final reward realizations. As has been discussed in the literature \cite{besson2017multi,bonnefoi2017multi,boursier2018sic}, {the no-sensing model} is  regarded as {a more difficult MP-MAB problem than the collision-sensing model}: not only are the players not allowed to communicate with each other, but they also cannot access the so-called collision indicator, i.e., whether there exist other players who pull the same arm.  In other words, all a player can do is, at each time slot, chooses an arm to pull and then observes a (final) reward signal from the arm. Existing literature has reported limited progress for no-sensing MP-MAB, and the known results typically have large regret gaps to the centralized MP-MAB with complete information.\footnote{{See Section \ref{sec:related} for a detailed literature review.}}

In this work, we study MP-MAB with collision-dependent reward distributions {and focus on the no-sensing model}. In particular, we claim the following contributions. 

\begin{itemize} [leftmargin=*]\itemsep=0pt

\item We propose a novel algorithm, called Error-Correction Collision Communication (EC3), that is developed with two novel tools that are traditionally outside the toolbox of multi-armed bandits research: reliable communication over noisy channels, and error-correction coding.
We show that these are powerful tools that lead to novel decentralized MP-MAB algorithms based on error-correction coding, and allow us to analyze their regret with the fundamental limit of the noisy communication channel.

\item  Under this framework, {we adopt error-correction coding to transmit sample reward means with adaptive quantization lengths between players}, which results in a communication regret that does not dominate the total regret when the coding rate is chosen properly. The asymptotic regret not only provides a strong performance guarantee for the new model with collision-dependent rewards, but also approaches the natural lower bound of centralized MP-MAB.

\item Experiments with practical error-correction codes on both synthetic and real-world datasets demonstrate the superiority of EC3. In particular, the results show that the choice of coding schemes has a profound impact on the regret performance.

\end{itemize}

The remainder of this paper is organized as follows. Section~\ref{sec:related} gives a brief summary of the related literature. The decentralized MP-MAB problem with the collision-dependent reward model is presented in Section~\ref{sec:prob}. We then first describe the implicit communication protocol design in Section~\ref{sec:info}, and then present the complete EC3 algorithm in Section~\ref{sec:alg}. Regret analysis is given in Section~\ref{sec:theory} and extensions are discussed in Section~\ref{sec:modelext}. Experimental results are reported in Section~\ref{sec:exp}, followed by the conclusion of the paper in Section~\ref{sec:conc}.

\section{Related Work}
\label{sec:related}
\textbf{Collision-sensing MP-MAB:} Decentralized stochastic MP-MAB problems, introduced by \cite{liu2010distributed} and \cite{anandkumar2011distributed}, disallow explicit communications among players. {The collision-sensing model is the most widely studied MP-MAB problem.} As the (single-player) stochastic MAB problem is well understood, a natural idea for collision-sensing stochastic MP-MAB is to {adopt mature single-player algorithms while avoiding} collisions for as much as possible. Examples include Explore-then-Commit \cite{rosenski2016multi}, UCB \cite{liu2010distributed,besson2017multi} and $\epsilon$-greedy \cite{avner2014concurrent}. Especially, the Musical Chairs algorithm \cite{rosenski2016multi} achieves {a constant regret with a high probability without any player pre-agreement (which can be easily turned into a logarithmic regret in expectation)}. {These algorithms all have sublinear regrets but these regrets are always  worse than the centralized MP-MAB regret (which is a natural lower bound for the decentralized setting) by a multiplicative factor $M$, where $M$ is the number of players. This gap is fundamental because the collision-avoidance approach is confined to letting players play $M$ separate single-player MAB games instead of one coordinated game.}

{As opposed to \emph{avoiding} collisions, recent advances in decentralized MP-MAB show that performance improvement can be achieved by purposely \emph{exploiting} collisions to communicate among players. One representative approach is the SIC-MMAB algorithm in \cite{boursier2018sic}, which is closely related to our work. It enables players to use forced collisions to communicate collected statistics and coordinately play the MAB game. With a careful coordination in the collision-sensing setting, SIC-MMAB leads to a regret that can approach the centralized lower bound. Similar ideas of using forced collisions are also utilized in \cite{avner2016multi,darak2019multi, tibrewal2019distributed} to communicate arm statistics or signal the need of arm switches with heterogeneous reward distributions for each player. More recent works have further explored and enhanced the collision-communication idea. For example, KL-UCB \cite{garivier2011kl} is incorporated in \cite{proutiere2019optimal}, and the heterogeneous setting is investigated in \cite{boursier2019practical, magesh2019multi}.  This idea has also been proved effective in the adversarial MP-MAB setting \cite{alatur2019multi}.}

\textbf{No-sensing MP-MAB:} {The no-sensing model, on the other hand, is much more challenging. Most state-of-the-art works, including Adapted-SIC-MMAB and SIC-MMAB2 in \cite{boursier2018sic} and EC-SIC in \cite{Shi2020aistats}, utilize the key idea of SIC-MMAB for the settings with zero rewards upon collisions. Specifically, Adapted-SIC-MMAB and SIC-MMAB2 transmit the same information repeatedly\footnote{This can be viewed as a primitive form of repetition coding, which will be discussed later.}, and utilize the unmistakable information from non-zero rewards (i.e., no collision) for reliable communication. Nevertheless, Adapted-SIC-MMAB has a dominating communication loss of order $O(\log(T)\log^2(\log(T)))$. To reduce the communication loss, SIC-MMAB2 only shares indices of accepted or rejected arms instead of arm statistics. However, its regret is still worse than the centralized lower bound by a multiplicative factor $M$. EC-SIC \cite{Shi2020aistats} uses the idea of coding but also strictly relies on the assumption of zero rewards upon collisions (which results in a much simplified communication scenario as discussed in Section \ref{sec:info}). Also, the proposed EC3 does not require prior knowledge of the sub-optimality gap, which is necessary for \cite{Shi2020aistats}. Meanwhile, some very recent works \cite{bubeck2020cooperative,bubeck2020coordination} have made progress on instance-independent regrets for the no-sensing setting, while the main focus of this paper is on instance-dependent regrets.

\textbf{Other related works:} There is also another line of research in MP-MAB where there are no collisions \cite{landgren2016distributed,landgren2018social,martinez2019decentralized,wang2019distributed}, but they are under a framework which is completely different from this work and thus are out of the scope of this paper.}

\section{MP-MAB Model and Problem Formulation}
\label{sec:prob}
In this section, we formulate the problem of decentralized no-sensing MP-MAB with collision-dependent rewards. The model contains a known number of arms $K$ but an unknown number of players $M\leq K$. The players are labeled from $1$ to $M$. The time horizon $T$ is known to the player. There is no {explicit} communication between players, but they are assumed to have a common knowledge of the time (e.g., shared clock for time synchronization). Without collisions, the reward $X_k^A{\geq 0}$ of arm $k\in[K]$ is sampled from distribution $A_k$ with $\mathbb{E}[X_k^A]=\mu[k]\in [0,1]$. Upon collisions, arm $k\in[K]$ is assumed to provide a reward $X_k^B{\geq 0}$ sampled from distribution $B_k$ with $\mathbb{E}[X_k^B]=\nu[k]\in[0,1]$. Both reward distributions, i.e., $A_k$ and $B_k$, are assumed to be $\sigma$-subgaussian.\footnote{This is a widely adopted assumption in MAB literature \cite{Bubeck:2012}. In particular, the commonly used assumption that the rewards are bounded in interval $[0,1]$ satisfies this property with $\sigma=\frac{1}{2}$.} {In the case of $P(X_k^B=0)=1$ for all arm $k\in[K]$, the model degenerates to the existing model with collision-eliminated rewards (i.e., zero reward for collision)} \cite{liu2010distributed,anandkumar2011distributed,rosenski2016multi,boursier2018sic}. {Under this framework}, the reward for player $m$ choosing arm $\pi_m(t)$ at time $t$ {can be expressed as}:
$$
r_m(t):=\underbrace{X^A_{\pi_m(t)}(t)(1-\eta_{\pi_m(t)}(t))}_{\text{no collision}}+\underbrace{X^B_{\pi_m(t)}(t)\eta_{\pi_m(t)}(t)}_{\text{collision}},
$$
where $\eta_{\pi_m(t)}(t)$ is the collision indicator defined by
$$
\eta_{k}(t):=\mathds{1}\{|C_{k}(t)|>1\},
$$
with $C_k(t):=\{m\in[M]|\pi_m(t)=k\}$.

In the \textit{collision-sensing} setting, both $r_m(t)$ and $\eta_{\pi_m(t)}(t)$ are known to the players. In the \textit{no-sensing} setting, only the final reward $r_m(t)$ is available to the players; no information on $\eta_{\pi_m(t)}(t)$ is available. With collision-dependent rewards, any reward within the support overlap of $A_k$ and $B_k$ can indistinguishably come from collision or non-collision. {Note that if the support of $A_k$ and $B_k$ do not overlap for all arms $k\in [K]$, then the no-sensing and collision-sensing problems are equivalent, since the reward itself conveys the collision information unambiguously.}

The following assumption is made on the MAB model to characterize the collision-dependent reward setting.\footnote{We note that several assumptions of this model, including Assumption \ref{aspt:bounds}, can be relaxed at the expense of unnecessarily complicating the notation, the proposed solution, and the analyses. These model extensions and assumption relaxations are deferred to Section~\ref{sec:modelext}, where the applicability becomes clear after the EC3 algorithm and its analysis are presented.}

	\begin{assumption}\label{aspt:bounds}
	A positive lower bound $\mu_{\min}$ and a {non-negative} upper bound $\nu_{\max}$ exist and are known to all players: $0{\leq} \max_{k\in[K]} \nu[k]\leq \nu_{\max}<\mu_{\min}\leq \min_{k\in[K]}\mu[k]$.
	\end{assumption}

Assumption \ref{aspt:bounds} formalizes the intuition that collisions reduce the mean reward (e.g., the average signal quality decreases if multiple users share the channel). {We immediately emphasize that this only applies to the {mean} reward, but not necessarily to the {\em realized} reward. The support of collision rewards and no-collision rewards may still overlap.} {Similar assumptions about $\mu_{\min}$ have been commonly used in the literature of no-sensing MP-MAB \cite{lugosi2018multiplayer,boursier2018sic} with zero rewards for collisions, where $\nu_{\max}$ is naturally $0$.}

The notation of regret can be generalized from the single-player setting \cite{Lai:1985} to MP-MAB, with respect to the best allocation of arms, as:
$$
\small
R(T):=T\sum_{m\in[M]}\mu_{(m)}-\mathbb{E} \left [\sum_{t=1}^T\sum_{m\in[M]}r_m(t) \right ],
$$
where $\mu_{(j)}$ is $j$-th ordered statistics of $\mu$, i.e. $\mu_{(1)}\geq \mu_{(2)}\geq\cdots\geq \mu_{(K)}$. A similar notion is also applied to the impaired rewards upon collisions as $\nu_{(j)}$. The goal of the players is to minimize the regret $R(T)$, which equivalently maximizes the obtained rewards.

The collision-dependent MP-MAB model {has not been well studied} for both collision-sensing and no-sensing {settings. However, with the perfectly known collision indicator in the collision-sensing setting, collision-dependent reward does not fundamentally change the problem from the model with zero rewards for collisions. Thus, we focus on explaining the algorithm design and the companion regret analysis for the more challenging no-sensing setting, and highlight the necessary modifications when applied to collision-sensing in Section \ref{subsec:sensing}.}

\section{Implicit Communication for Collision-Dependent MP-MAB}\label{sec:info}

One of the core ingredients of EC3 is the implicit communication protocol via purposely instigated collisions. We detail this idea for the collision-dependent MP-MAB model, and then present the overall EC3 in the next section, in which this communication protocol is a key component.

\subsection{Implicit Communication}

The implicit communication protocol for collision-dependent MP-MAB is a generalization of the methods in \cite{boursier2018sic,Shi2020aistats}, and is illustrated in Fig.~\ref{fig:comexamp} for a two-player two-arm example. The key idea is that both players, based on some prior synchronization mechanism to be explained later, take predetermined turns (Player 1 first transmits, and then Player 2 transmits) to communicate by having the ``receive'' player sample her own communication arm ({without loss of generality we assume the same index of player and her communication arm, i.e., arm $m$ for player $m$}) and the ``send'' user either pull (create collision; bit $1$) or not pull (create no collision; bit $0$) receiver's communication arm to convey one bit information. Doing so sequentially over a fixed duration of time slots would convey a bit sequence $\boldsymbol{b}$ of the same length between the two players, \textit{if the collision indicator is perfectly known}. For a player that is not engaged in the current peer-to-peer communication, she keeps pulling her communication arm to avoid interrupting other ongoing communications. The general send and receive protocols {for an arbitrary player $m$ to communicate a message $\boldsymbol{b}$ of length $l_b$ to player $n$ are given in Algorithm \ref{alg:send} and Algorithm \ref{alg:receive}, respectively.} In this way, players can share individual arm statistics with each other so that they can explore the arms cooperatively.

\begin{figure}[thb]
	\vspace{-0.1in}
	\setlength{\abovecaptionskip}{0.2pt}
	\centering
	\includegraphics[width=0.8 \linewidth]{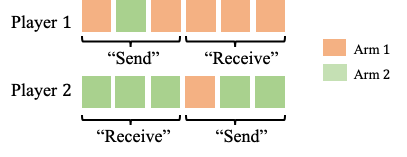}
	\caption{Illustration of implicit communication for MP-MAB.}
	\label{fig:comexamp}
\end{figure}

\begin{algorithm}[thb]
	\small
	\caption{\texttt{Send}}
	\label{alg:send}
	\begin{algorithmic}[1]
		\Require sender index $m$, receiver index $n$, bit sequence to send $\boldsymbol{b}$
		\State $l_b\gets$length($\boldsymbol{b}$)
		\For{$l=1,...,l_b$}
		\If{$\boldsymbol{b}[l]=1$}
		\State Pull arm $n$\Comment{\textit{Collision: send bit $1$}}
		\Else
		\State Pull arm $m$\Comment{\textit{No collision: send bit $0$}}
		\EndIf
		\EndFor
	\end{algorithmic}
\end{algorithm}

\begin{algorithm}[thb]
	\small
	\caption{\texttt{Receive}}
	\label{alg:receive}
	\begin{algorithmic}[1]
		\Require receiver index $n$
		\For{$l=1,...,l_b$}
		\State Pull arm $n$\Comment{\textit{Pull her own communication arm}}
		\If{$\eta_{n}(t)=1$}
		\State {$\hat{\boldsymbol{b}}[l]\gets 1$}\Comment{\textit{Collision: receive bit $1$}}
		\Else
		\State {$\hat{\boldsymbol{b}}[l]\gets 0$}\Comment{\textit{No collision: receive bit $0$}}
		\EndIf
		\EndFor
		\Ensure  {received bit sequence $\hat{\boldsymbol{b}}$} 
	\end{algorithmic}
\end{algorithm}

\subsection{Reliable Communication with Error-Correction coding}\label{subsec:rel_comm}

In the no-sensing setting\footnote{{Some no-sensing papers \cite{bubeck2020cooperative,bubeck2020coordination} strictly assume no implicit communications. This work, however, does not make such assumption and focuses on leveraging implicit communications even in the no-sensing setting.}}, collisions cannot be perceived without error, which poses a challenge for implicit communication. In the previous literature where collisions always lead to zero rewards, any observed reward that is strictly larger than $0$ can be received as bit $0$ \textit{error-free}, which means that error can happen in only one direction ($0$-to-$1$ error) \cite{Shi2020aistats}. However, this is no longer true for collision-dependent rewards, where \textit{two-way errors} between bit $1$ and $0$ may happen. {This generalization results in that previous no-sensing solutions proposed in \cite{boursier2018sic,Shi2020aistats} cannot be trivially extended to the current collision-dependent reward model.} 

Luckily, having an error-free communication is not the ultimate goal -- the objective is to have sufficiently reliable communication to share arm statistics so that the regret caused by communication errors does not dominate the overall regret. With this in mind, we model no-sensing implicit communication as the \textit{reliable communication over noisy channels} problem \cite{CoverBook}, and leverage its fundamental limit \cite{gallager1968information} to establish the optimal regret scaling. Specifically, communications on arm $k$ is modeled as transmission over a binary channel in Fig.~\ref{fig:subgaussian}: bit $0$ (no collision) and bit $1$ (collision) are received as samples from two distinct $\sigma$-subgaussian distributions $A_k$ and $B_k$. More precise channel models can be used if additional knowledge of the distribution is available. For example, the additive white Gaussian noise (AWGN) channel is shown in Fig.~\ref{fig:awgn}, in which the noise follows a Gaussian distribution with mean zero and variance $\sigma^2$.

\begin{figure}[htb]
    \centering
    \setlength{\abovecaptionskip}{-0.1pt} 
	\begin{minipage}[t]{0.49\linewidth}
		\centering
		\includegraphics[width=0.9\textwidth]{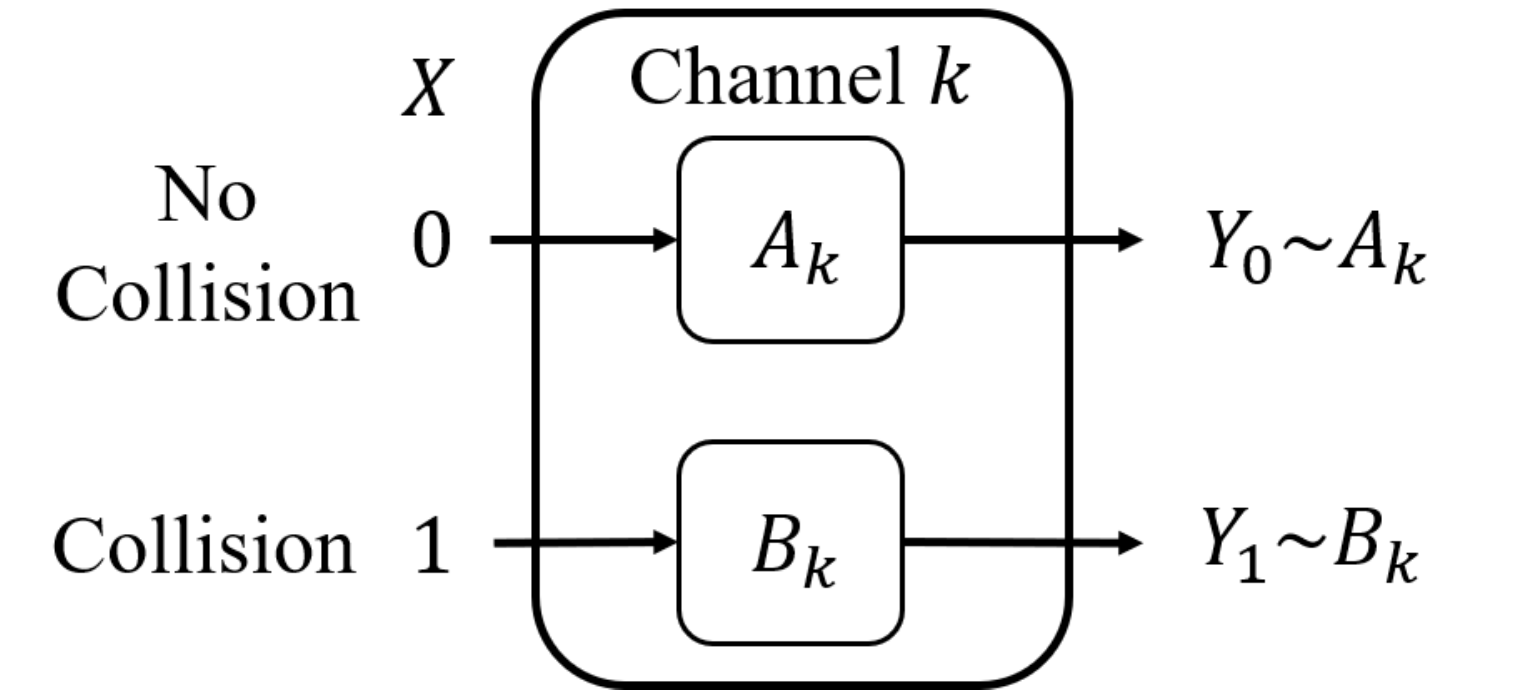}
		\caption{Binary channel}
		\label{fig:subgaussian}
	\end{minipage}
	\begin{minipage}[t]{0.49\linewidth}
		\centering
		\includegraphics[width=0.7\textwidth]{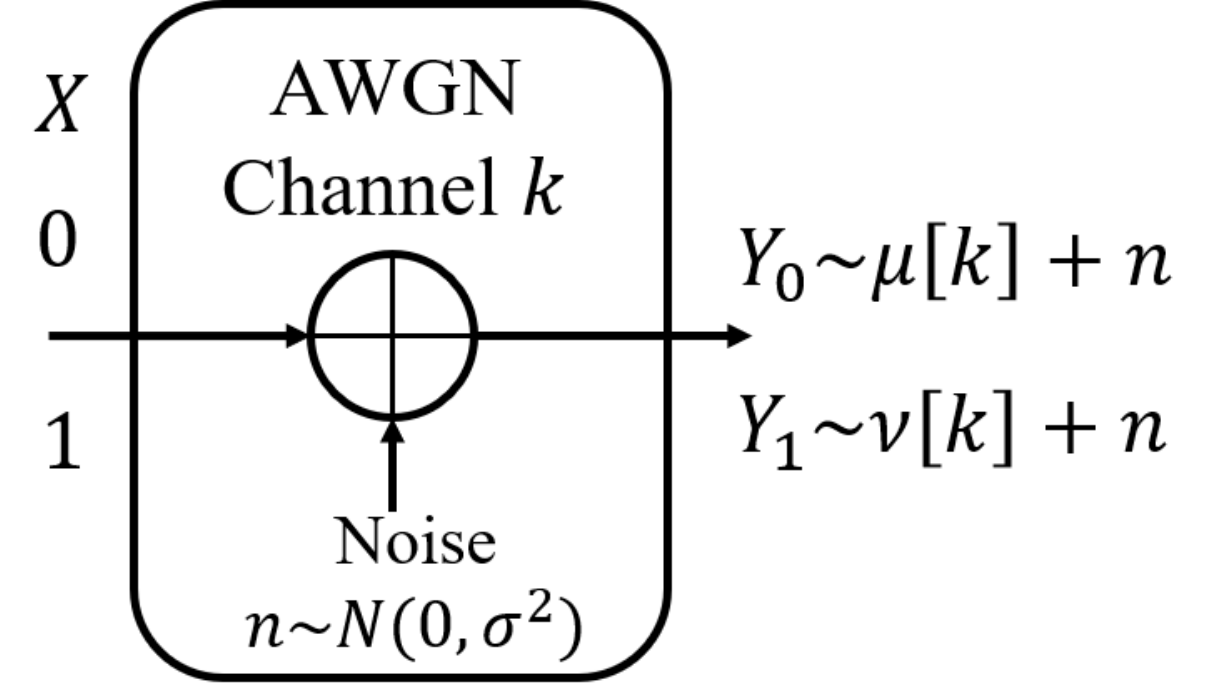}
		\caption{AWGN channel}
		\label{fig:awgn}
	\end{minipage}
	\vspace{-0.15in}
\end{figure}

To achieve reliable communications when transmissions are not error-free, coding is an effective technique. However, directly working on a particular coding scheme would lead to regret analysis that depends on the choice of the code, and raise the question whether a different choice may improve the result. We thus first approach the algorithm design from the optimal scheme that no code can beat, and then evaluate practical coding choices.  An important result of channel coding for communication over a noisy channel, known as the \textit{random coding error exponent} \cite{gallager1968information}, is fundamental to our problem.

\begin{theorem}\label{thm:error_exponent}
	For the binary memoryless channel described in Fig.~\ref{fig:subgaussian}, if the coding rate $R_c<C$ where $C$ is the channel capacity, there exists a code of block length $N$ without feedback such that the error probability is bounded by
	\begin{equation}\label{eqn:error_exponent}
	    P_e\leq \exp[-N {E_r(R_c)}],
	\end{equation}
	where {$E_r(R_c)$} is the random coding error exponent that depends on the coding rate $R_c$ and the channel model.
\end{theorem}

\begin{corollary}\label{col:coding_length}
    For the binary memoryless channel described in Fig.~\ref{fig:subgaussian}, denoting $N'(L)$ as the block length of the optimal code  without feedback that can transmit a message of length $L$ with an error probability $P_e\leq {1}/{T}$, it holds that  $N'(L)\leq \max\left\{\frac{L}{C-\varphi}, \frac{\log(T)}{E_r(C-\varphi)}\right\}$, where $\varphi \in (0,C)$ is an arbitrary constant.
\end{corollary}

Corollary \ref{col:coding_length} suggests that to transmit an $L$-bit message over the binary channel, there exists an optimal coding scheme with length $N'(L)=O(\log(T))$ to achieve an error rate of $O\left({1}/{T}\right)$. This scaling of the error rate is critical for the subsequent regret analysis.  The detailed proof of Theorem \ref{thm:error_exponent} is in \cite{gallager1968information} while the proof of Corollary \ref{col:coding_length} can be found in Appendix \ref{app:coding_length}. For each arm $k$, we have a corresponding channel capacity $C_k$ characterized by $\mu_{k}$ and $\nu_k$. To ease the exposition and highlight the key algorithm design and regret analysis, we use Assumption~\ref{aspt:bounds} to provide a universal description of the worst-case $\mu_{\min}$ and $\nu_{\max}$ for all channels (i.e., arms), instead of using the arm-specific characterization. As a result, a common (worst-case) channel capacity $C$ and the corresponding error exponent $E_r(\cdot)$ can be derived that depend on $\mu_{\min}$ and $\nu_{\max}$. This universal description further leads to a unified upper bound of $N'(\cdot)$ as shown in Corollary~\ref{col:coding_length}, which is useful in the analysis. We emphasize that this simplification of Assumption~\ref{aspt:bounds} is not fundamental to the algorithm design and regret analysis; detailed discussions on its relaxation and corresponding enhancement can be found in Section~\ref{subsec:ext_aspt1}. In reality, different finite block-length codes may have different error-rate performances and {require different coding lengths to achieve an error rate smaller than ${1}/{T}$}; several practical codes are discussed in Section \ref{sec:modelext} and \ref{sec:exp}. Correspondingly, general functions\footnote{{Note that for the receiving protocol (i.e., Algorithm \ref{alg:receive}) in the no-sensing algorithm, players would use the reward values instead of the collision indicator to determine the received bits. Details can be founded in Appendix \ref{app:code} where specific coding schemes are discussed.}} \texttt{Send()}, \texttt{Receive()}, \texttt{Encoder()} and \texttt{Decoder()} are used in the algorithm descriptions without specifying the actual coding schemes.

\section{EC3 Algorithm Design}
\label{sec:alg}

The EC3 algorithm is compactly described in Algorithm~\ref{alg:overall}. It can be divided into four phases: initialization phase, exploration phase, communication phase, and exploitation phase, as illustrated in Fig.~\ref{fig:ec3struct}. The purpose of the initialization phase is to estimate the unknown $M$ with a sequential hopping protocol, for all players. Thus, if $M$ is known a priori, this phase can be skipped entirely. After the initialization phase, the EC3 iterates between exploration and communication phases, whose details are presented subsequently, {until all the optimal $M$ arms have been found and are allocated to players. After the allocation is made, each player fixates on her allocated arm and enters the exploitation phase where she always plays this arm until the end of $T$, with no further communication.}

\begin{figure}[thb]
	\vspace{-0.05in}
	\setlength{\abovecaptionskip}{0.2pt}
	\centering
	\includegraphics[width=0.8 \linewidth]{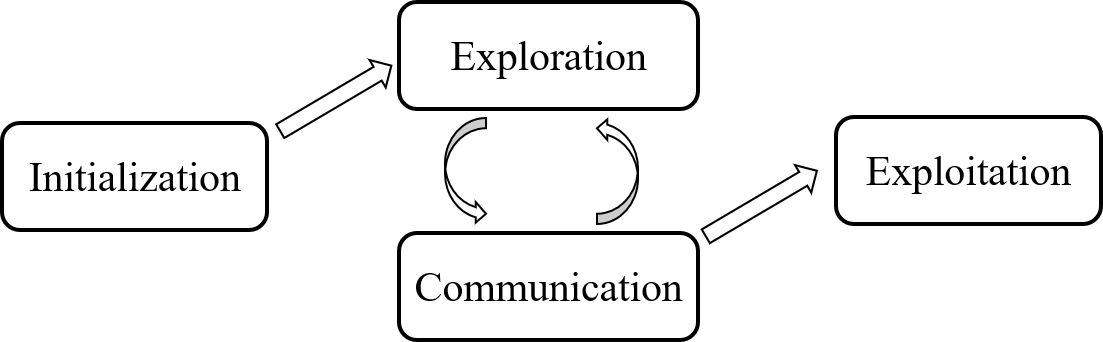}
	\caption{Algorithmic structure of EC3.}
	\label{fig:ec3struct}
	\vspace{-0.10in}
\end{figure}

It is essential to note the importance of synchronization in EC3, i.e., all players enter each phase at the same time (or at least with high probability). Also, in order to improve the communication efficiency, a server-client structure is adopted such that player $1$ serves as the leader and other players as followers.

\begin{algorithm}[thb]
	\small
	\caption{The EC3 Algorithm}
	\label{alg:overall}
	\begin{algorithmic}[1]
		\Require $T$, $K$, $\mu_{\min}$, $\nu_{\max}$, $\sigma$, player index $m$
		\State Initialize $p \gets 1$; active arms set $\mathcal{K}_p \gets [K]$; exploration sequence $E_p^m\gets \{m,m+1,...,K,1,...,m-1\}$; exploration time $T_{p-1}\gets 0$; accepted arms set $\mathcal{A}\gets \emptyset$;  rejected arms set $\mathcal{R}\gets\emptyset$
		\Statex \LeftComment{\textit{Initialization phase}}
		\State $M\gets$ \texttt{Estimate\_M} (); $M_p\gets M$
		\While{$|E_p^m|>1$} 
		\Statex \LeftComment{\textit{Exploration phase}}
		\State Pull arms following the order of the exploration sequence $E_p^m$  each for $2^p \lceil \sigma^2\log(T) \rceil$ times
		\State Update sample mean $\hat{\mu}_p^m[k], \forall k\in \mathcal{K}_p$
		\Statex \LeftComment{\textit{Communication phase}}
		\State $T_p \gets T_{p-1}+M_p2^p\lceil \sigma^2\log(T) \rceil $; $B_{T_p} \gets \sqrt{\frac{2\sigma^2\log(T)}{T_p}}$
		\State $Q_p \gets \left\lceil\log_2\left(\frac{1}{B_{T_p}}\right)\right\rceil$; $L_p \gets 1+Q_p$	
		\If {$m=1$}
		\State $(M_{p+1}, \mathcal{K}_{p+1}, E^1_{p+1})\gets\texttt{Leader} ()$
		\Else 
		\State $(M_{p+1}, \mathcal{K}_{p+1}, E_{p+1}^m)\gets\texttt{Follower} ()$
		\EndIf
		\State $p\gets p+1$
		\EndWhile
		\Statex \LeftComment{\textit{Exploitation phase}}
		\State Pull the only element in $E_p^m$ until the end of $T$ time slots 
	\end{algorithmic}
\end{algorithm}

\subsection{Initialization Phase}
The initialization phase takes the spirit from \cite{boursier2018sic} and is presented here for completeness. It consists of two parts. In the first part, each player $m$ fixates on arm $m$ and sends a coded bit $1$ to the leader on arm $1$ sequentially, and the leader keeps receiving potential bits $1$ from arm $2$ to arm $K$. Then, the leader transmits the estimated $M$ to all the followers. After this phase, all the players are aware of the number of players $M$ and Lemma~\ref{lem:init_regret} guarantees that this estimation is correct with a high probability. The initialization algorithm (\texttt{Estimate\_M}) for the leader is shown in Algorithm \ref{alg:init}; refer to \cite{boursier2018sic} for more details.

\begin{algorithm}[thb]
	\caption{\texttt{Estimate\_M (Leader)}}
	\small
	\label{alg:init}
	\begin{algorithmic}[1]
		\State Initialize $M\gets 1$
		\For {$k=2,...,K$}
		\State $G\gets\texttt{Decoder}(\texttt{Receive}(k, \text{coded bit $1$ or $0$ signal}))$
		\Statex\Comment{\textit{Receive signal from the potential player on arm $k$}}
		\If {$G=1$}
		\State $M\gets M+1$
		\EndIf
		\EndFor
		\For {$m=2,...,M$}
		\State \texttt{Send}$(1, m, \texttt{Encoder}(M))$\Comment{\textit{Send the estimated $M$}}
		\EndFor
	\end{algorithmic}
\end{algorithm}

\subsection{Exploration Phase}
{In the exploration phases, players are categorized as active or inactive. The \textit{active players} are those who have not been allocated with a specific arm yet, and they are in charge of exploration. The \textit{inactive players} have already been allocated with previously declared optimal arms, and they fixate on those arms. In the $p$-th exploration phase, we define $\mathcal{M}_p$ as the set of active players, whose cardinality is $M_p$. Similarly, \textit{active arms} $\mathcal{K}_p$, with cardinality $K_p$, are defined as the set of arms that have not been decided to be optimal or sub-optimal yet.}

{In phase $p$, the active arms are ordered according to the player index $m$ in the exploration sequence $E_p^m$ for all the active players $m\in \mathcal{M}_p$. Following the order in $E_p^m$, active players sequentially hop among the active arms for total $K_p2^p\lceil\sigma^2\log(T)\rceil$ time steps.  Thus, each active arm is pulled $2^p\lceil\sigma^2\log(T)\rceil$ times by each active player. For inactive players, their exploration sequences are of the same length as the ones for active players but consist of only the arm indices on which they have fixated.} With player synchronization {and distinct exploration sequences}, these explorations are collision-free and sample means of no-collision distributions are collected by the active players. {We denote the sample mean collected by player $m$ for arm $k$ after phase $p$ as $\hat{\mu}_{p}^m[k] $. }

{We note that the choice of exploration length $K_p2^p\lceil\sigma^2\log(T)\rceil$ in EC3 is expanded by a factor of $\lceil\sigma^2\log(T)\rceil$ compared to SIC-MMAB in \cite{boursier2018sic}. This expansion is crucial in achieving a sublinear regret in the no-sensing setting, as it decreases the overall number of the communication phases (see Section~\ref{subsec:comm}) while maintaining their efficiency. With this exploration protocol, any active arm $k\in [K_p]$ is pulled $T_p = \sum_{m=1}^M T^p_m = \sum_{q=1}^p M_q2^q\lceil\sigma^2\log(T)\rceil$ times by all the players up to phase $p$, where $T^p_{m}$ denotes the number of pulls that player $m$ (active or inactive) has performed on this active arm $k$ up to phase $p$.}

\begin{algorithm}[htb]
	\small
	\caption{\texttt{Leader}} 
	\label{alg:comm_leader}
	\begin{algorithmic}[1]
		\Statex \LeftComment{\textit{Receive and update arm statistics}}
		\For {$i=2,...,M$}
		\State $\forall k\in \mathcal{K}_p, \bar{\mu}_p^i[k]\gets \texttt{Decoder}(\texttt{Receive}(i, \text{coded } \bar{\mu}_p^i[k]))$
		\EndFor
        \State $\forall k\in \mathcal{K}_p, \bar{\mu}_p[k]\gets \sum_{i=1}^M\bar{\mu}^i_p[k]\cdot T_p^i/T_p$
		\Statex \LeftComment{\textit{Arm acceptation and rejection}}
		\State  $\mathcal{R}_p\gets$ all arms $k\in \mathcal{K}_p$ satisfying $$|\{j\in\mathcal{K}_p|\bar{\mu}_p[j]-2B_{T_p}\geq \bar{\mu}_p[k]+2B_{T_p}\} |\geq M_p$$ 
		\State $\mathcal{A}_p\gets$ all arms $k\in \mathcal{K}_p$ satisfying $$| \{j\in\mathcal{K}_p|\bar{\mu}_p[k]-2B_{T_p}\geq \bar{\mu}_p[j]+2B_{T_p}\} |\geq K_p-M_p$$
		\Statex \LeftComment{\textit{Transmit accepted and rejected arms}}
		\State \texttt{Send}$(1, i, \texttt{Encoder}\left(\{|\mathcal{A}_p|,|\mathcal{R}_p|\}\right))$ for $i=2,..., M$
		\State \texttt{Send}$(1, i,  \texttt{Encoder}(\{\mathcal{A}_p, \mathcal{R}_p\}))$ for $i=2,..., M$
		\Statex \LeftComment{\textit{Determine exploration set for the next phase}}
		\State $\mathcal{A}\gets \mathcal{A}\cup\mathcal{A}_p$; $\mathcal{R}\gets \mathcal{R}\cup\mathcal{R}_p$; $\mathcal{K}_{p+1}\gets \mathcal{K}_p\backslash(\mathcal{A}_p\cup\mathcal{R}_p)$
		\If {$M\leq |\mathcal{A}|$}  
		\State $E_{p+1}^1 \gets \{\mathcal{A}[M]\}$
		\Else \State $M_{p+1}\gets M-|\mathcal{A}|$
		\State $E_{p+1}^1 \gets \{\mathcal{K}_{p+1}[1],...,\mathcal{K}_{p+1}[K_{p+1}]\}$
		\EndIf
	\end{algorithmic}

\end{algorithm}
\begin{algorithm}[htb]
	\small
	\caption{\texttt{Follower}} 
	\label{alg:comm_follower}
	\begin{algorithmic}[1]
		\Statex \LeftComment{\textit{Transmit arm statistics}}
		\State Quantize $\hat{\mu}_p^m[k]$ by $L_p$ bits into $\bar{\mu}_p^m[k]$, $\forall k\in \mathcal{K}_p$
		\State $\forall k\in \mathcal{K}_p$, \texttt{Send}$(m, 1, \texttt{Encoder}(\bar{\mu}_p^m[k]))$
		\Statex \LeftComment{\textit{Receive accepted and rejected Arms}} 
		\State $\{N_{\text{a}},N_{\text{r}}\}\gets \texttt{Decoder}(\texttt{Receive}(m, \{\text{coded } |\mathcal{A}_p|,|\mathcal{R}_p|\}))$
		\State $\{\mathcal{A}_p,\mathcal{R}_p\}\gets\texttt{Decoder}(\texttt{Receive}(m, \{\text{coded } \mathcal{A}_p,\mathcal{R}_p\}))$
		\Statex \LeftComment{\textit{Determine exploration set for the next phase}}
		\State $\mathcal{A}\gets \mathcal{A}\cup\mathcal{A}_p$; $\mathcal{R}\gets \mathcal{R}\cup\mathcal{R}_p$; $\mathcal{K}_{p+1}\gets \mathcal{K}_p\backslash(\mathcal{A}_p\cup\mathcal{R}_p)$
		\If {$M-m+1 \leq |\mathcal{A}|$} 
		\State $E_{p+1}^m \gets \{\mathcal{A}[M-m+1]\}\times K_{p+1}$
		\Else \State $M_{p+1}\gets M-|\mathcal{A}|$ 
		\State $E_{p+1}^m \gets \{\mathcal{K}_{p+1}[m],...,\mathcal{K}_{p+1}[K_{p+1}],\mathcal{K}_{p+1}[1], ..., \mathcal{K}_{p+1}[m-1]\}$
		\EndIf
	\end{algorithmic}
\end{algorithm}

\subsection{Communication Phase}\label{subsec:comm}
After each exploration phase, all the players, including the inactive ones, share arm statistics via collisions in the communication phase. The subroutines for communication are described in Algorithms~\ref{alg:comm_leader} and \ref{alg:comm_follower} for the leader and followers, respectively. The followers encode and send their collected statistics {of active arms}, i.e., the sample mean of collected rewards $\hat{\mu}^m_p[k]$, to the leader. The leader decodes the messages and uses the information from all the players to determine the sets of accepted and rejected arms, and then sends back these two sets to the followers. Note that all of these pairwise communications utilize the implicit communication protocol described in Section \ref{sec:info}.\footnote{As mentioned in Section \ref{sec:info}, when a player is involved in neither transmission or reception, she keeps sampling her own communication arm. This ``idle'' protocol is absorbed in the \texttt{Send()} and \texttt{Receive()} functions in Algorithms \ref{alg:comm_leader} and \ref{alg:comm_follower}.}

Another important feature, besides the error-correction coding, is the careful algorithmic choice to control the communication regret at or below the order of regret caused by the exploration phases. This is achieved by transmitting the quantized sample means, denoted as $\bar{\mu}^m_p[k]$ for sample mean $\hat{\mu}^m_p[k]$, with an adaptive quantization length. This adaptive quantization length is selected such that the quantization error does not dominate the current sampling uncertainty. Specifically, by adopting $B_{T_p} = \sqrt{\frac{2\sigma^2\log(T)}{T_p}}$ as the confidence bound for the sample means at phase $p$ (which will be proved to be effective in the regret analysis), the quantization length can be chosen as $L_p = 1+Q_p$, with $1$ bit representing the integer part and $Q_p = \lceil\log_2(\frac{1}{B_{T_p}})\rceil$ bits representing the decimal part. With this quantization length, the quantization error satisfies $2^{-Q_p}\leq B_{T_p}$, which does not dominate the sampling uncertainty. In other words, with non-dominating quantization errors, sample means can be reconstructed at the leader with approximately the same accuracy.

After receiving quantized sample means from followers, arm acceptation and rejection are performed by the leader {with the aggregated sample means $\bar{\mu}_p[k]=\frac{1}{T_p}\sum_{m=1}^M\bar{\mu}^m_p[k]T^m_p$}. An arm $k$ is accepted (added to the accepted set $\mathcal{A}_p$) if it is among the top-$M_p$ active arms with high probability, which is characterized by
$$\left | \{j\in \mathcal{K}_p|\bar{\mu}_p[k]-{2}B_{T_p}\geq \bar{\mu}_p[j]+{2}B_{T_p}\} \right |\geq K_p-M_p.$$ 
Similarly, arm $k$ is rejected (added to the rejected set $\mathcal{R}_p$) if
$$|\{j\in  \mathcal{K}_p|\bar{\mu}_p[j]-{2}B_{T_p}\geq \bar{\mu}_p[k]+{2}B_{T_p}\} |\geq M_p,$$
which means that there are at least $M_p$ active arms that are better than $k$ with a high probability. Note that the overall confidence bound used in arm acceptation and rejection is $2B_{T_p}$, which consists of the sampling uncertainty and quantization error. The analysis in Section \ref{sec:theory} shows that the designed adaptive quantization length maintains a high probability of success for accepting and rejecting arms.

The set of accepted and rejected arms are then sent back to the followers, who then determine the exploration sequence for the next phase, i.e., $E_{p+1}^m$. Once there are at least $M$ accepted arms, players begin the exploitation phase and no longer communicate afterwards.

\section{Theoretical Analysis}
\label{sec:theory}
This section provides a theoretical analysis of the proposed algorithm. The overall regret for the EC3 algorithm can be decomposed as $R(T)=R^{\text{init}}+R^{\text{expl}}+R^{\text{comm}}$. The first, second and third terms respectively refer to the regret caused by the initialization, exploration and communication phases. The overall regret of EC3 in the no-sensing setting is given by Theorem \ref{thm:overall_nosens}. Each component regret is subsequently analyzed. The following proofs are presented under the no-sensing setting, which can be easily converted to the collision-sensing setting discussed in Section \ref{subsec:sensing}. {Note that in the analysis, $\sigma$ is assumed to be $1$ in order to give a better illustration of the regret; the proofs can be easily extended to any arbitrary value of $\sigma$.}
\begin{theorem}[\textbf{No-sensing regret}]
\label{thm:overall_nosens}
    The regret of the no-sensing EC3 algorithm is upper bounded as
	\begin{equation*}
	\small
	\begin{aligned}
    R_{\text{ns}}&(T)\leq 113\sum_{k>M}\frac{8\sqrt{6}\log(T)}{\mu_{(M)}-\mu_{(k)}} \\
    &+2\log_2 \left(\frac{8\sqrt{6}}{\Delta} \right)M^2K{\Delta_c}N'\left(L_H\right)\\
    &+\left(4M^2\log_2 \left(\frac{8\sqrt{6}}{\Delta} \right)+2MK+M^2K\right){\Delta_c}N'\left(\lceil\log_2(K)\rceil\right)\\
    &+6M^2K{\Delta_c}\log_2(T),
	\end{aligned}
	\end{equation*}
	where  {$\Delta=\mu_{(M)}-\mu_{(M+1)}$, $N'(L)=O(\log(T))$ is defined in Corollary \ref{col:coding_length}, {$\Delta_c = \mu_{(1)}-\nu_{(K)}$} and   $L_H=1+\left\lceil\log_2\left(\frac{8\sqrt{3}}{\Delta}\right)\right\rceil$}. Furthermore, the instance-independent regret of the no-sensing EC3 algorithm is upper bounded as
	\begin{equation*}\small
		R_{\text{ns}}(T)\leq O\left(K\sqrt{T\log(T)}\right).
	\end{equation*}
\end{theorem}
{We note that the parameter $\Delta_c\leq 1$ in Theorem~\ref{thm:overall_nosens} represents a general upper bound for collision loss per step, e.g., for communications, and discussions on a more precise arm-specific characterization can be found in Section \ref{subsec:ext_aspt1}.}
\begin{corollary}\label{col:overall_nosens}
	The asymptotic upper bound of the no-sensing EC3 algorithm can be derived as
	\begin{equation}\label{eqn:asym2}
		\small
		R_{\text{ns}}(T){\leq}O\left(\sum_{k> M}\frac{\log(T)}{\mu_{(M)}-\mu_{(k)}}+\frac{M^2K{\Delta_c}\log(\frac{1}{\Delta})\log(T)}{{E_r(C-\varphi)}}\right),%
	\end{equation}
	{where $\varphi$ is an arbitrary constant in $(0,C)$.}
\end{corollary}

To understand the upper bounds in Theorem \ref{thm:overall_nosens} and Corollary \ref{col:overall_nosens}, it is instrumental to compare the regret with the lower bound.
As proven in \cite{anantharam1987asymptotically}, the performance of any consistent algorithm $\phi$ for the corresponding \textit{centralized} MP-MAB model satisfies: 
\begin{equation}\label{eqn:lower_bound}
	\liminf_{T\to\infty}\frac{R^{\phi}(T)}{\log(T)}\geq \sum_{k> M}\frac{\mu_{(M)}-\mu_{(k)}}{\text{kl}(\mu_{(k)},\mu_{(M)})},
\end{equation}
where $\text{kl}(\mu_{(i)},\mu_{(j)})$ denotes the KL-divergence between the two corresponding reward distributions. Eqn.~\eqref{eqn:lower_bound} naturally serves as a {lower bound} for the decentralized MP-MAB problems with collision-dependent rewards.  We note that the exploration loss is of the same order\footnote{In the case that reward distributions $\{A_k\}$ are Bernoulli, it can be observed that the two terms are of the same order by implementing the inequality $\text{kl}(\mu_i,\mu_j)\geq 2(\mu_i-\mu_j)^2$.} as Eqn. \eqref{eqn:lower_bound} and only an additional regret term in Eqn.~\eqref{eqn:asym2} is introduced, referring to the initialization and communication regret caused by the nature of decentralization and no-sensing. Hence, we conclude that EC3 algorithm for the no-sensing setting is order-optimal.

We also note that compared with collision-eliminated rewards \cite{boursier2018sic,Shi2020aistats}, collision-dependent rewards have the following effects on the regret bound in Theorem~\ref{thm:overall_nosens}. First, as stated in Section~\ref{subsec:rel_comm}, the coding length $N'(\cdot)$ is characterized by the worst-case $\mu_{\min}$ and $\nu_{\max}$, which is unique in the collision-dependent reward setting. Also, the parameter $\Delta_c$ indicates that although collisions happen during communications, some reduced rewards can still be collected (instead of zero reward in the collision-eliminated reward setting).

The remainder of this section is dedicated to the proof of Theorem~\ref{thm:overall_nosens}. To accomplish this, we first define the \textit{typical event} as the success of initialization, communication and exploration throughout the entire horizon $T$. More specifically, we define three events: 
\begin{align*}
	&A_1=\{\text{all players estimate $M$ correctly}\};\\ 
	&A_2=\{\text{all communication messages are decoded correctly}\};\\
	&A_3=\{|\bar{\mu}_p[k]-\mu[k]|\leq {2}B_{T_p} \text{ holds}, \forall k\in\mathcal{K}_p, \forall p \}.
\end{align*}
We use $P_s$ to denote the probability that the typical event, which is $A_1 \cap A_2 \cap A_3$, happens. The regret caused by the \textit{atypical event} can be simply upper bounded by a linear regret $MT\Delta_c$. Then the result of Theorem \ref{thm:overall_nosens} can be proved by controlling $P_s$ to balance both events. We give the detailed proof for the no-sensing regret below. In addition, the proofs of all lemmas, except Lemma~\ref{lem:expl_regret}, can be found in Appendix~\ref{app:lem:init_regret} to \ref{app:lem:comm_regret}.

\subsection{Initialization Regret}

The initialization phase can be viewed as $2K$ communications; hence the regret can be bounded as follows. 
\begin{lemma}
\label{lem:init_regret}
	Denote $P_i = \mathbb{P}\{A_1\}$. Then $P_i\geq 1-\frac{M+K}{T}$, and the regret of initialization phase is bounded as:
	$$
	R^{\text{init}}\leq  2MKN'({\lceil\log_2(K)\rceil}){\Delta_c}.
	$$
\end{lemma}

\subsection{Exploration Regret}
The exploration regret is characterized as follows.
\begin{lemma}\label{lem:expl_regret}
\begin{enumerate}[leftmargin=12pt,topsep=0pt, itemsep=0pt,parsep=0pt]
\item[1)] 
Denote $P_s = \mathbb{P}\{A_1 \cap A_2 \cap A_3\}$ as the probability of the typical event.  Then {it holds that} $$P_s \geq 1-\frac{6MK\log_2(T)}{T}.$$
\item[2)] 
Conditioned on the typical event, the regret of exploration phase is bounded as:
	\begin{equation*}\small
		R^{\text{expl}} \leq 113\sum_{k>M}\min \left\{\frac{8\sqrt{6}\log(T)}{\mu_{(M)}-\mu_{(k)}},\sqrt{T\log(T)}\right\}.
	\end{equation*}
\end{enumerate}
\end{lemma}

\begin{proof}
Based on the optimal error-correction code (in the sense of error exponent) in {Corollary~\ref{col:coding_length}}, we first show that event $A_2$ happens with a high probability. 

\begin{lemma}\label{lem:success_com}
	Denoting $P_r=\mathbb{P}\{A_2\}$, we have $$P_r \geq 1-\frac{3MK\log_2(T)}{T}.$$
\end{lemma}

Next we analyze event $A_3$, which indicates the estimations of all arms reconstructed by the leader are within the confidence interval of $2B_{T_p}$ at all phases, in the following lemma.

\begin{lemma}\label{lem:success_est}
	Denoting $P_c=\mathbb{P}\{A_3\}$, it can be bounded as 
	$$P_c\geq 1-\frac{2MK\log_2(T)}{T}.$$
\end{lemma}
 
 Combining $P_i$, $P_r$ and $P_c$, the probability $P_s$ as defined in Theorem \ref{lem:expl_regret} can be obtained and the number that an arm is pulled before being accepted or rejected is controlled as:
 \begin{lemma}
 	\label{lem:expl_pull}
 	Conditioned on the typical event, every optimal arm $k$ (i.e., the top-$M$ arms) is accepted after at most $\frac{384\log(T)}{(\mu[k]-\mu_{(M+1)})^2}$ pulls, and every sub-optimal arm $k$ is rejected after at most $\frac{384\log(T)}{(\mu_{(M)}-\mu[k])^2}$ arm pulls. 
 \end{lemma}

Since the exploration phases are collision-free in the typical event, the exploration regret can be decomposed as \cite{anantharam1987asymptotically}:
	\begin{equation}
	\label{eqn:expl_decom}
	\begin{aligned}
	R^{\text{expl}}&=\sum_{k>M}\left(\mu_{(M)}-\mu_{(k)}\right)T^{\text{expl}}_{(k)}(T)
	\\&+\sum_{k\leq M}\left(\mu_{(k)}-\mu_{(M)}\right)\left(T^{\text{expl}}-T^{\text{expl}}_{(k)}(T)\right),
	\end{aligned}
	\end{equation}
	where $T^{\text{expl}}$ is the overall time of exploration and exploitation phase and $T^{\text{expl}}_{(k)}(T)$ is the number of time steps where the $k$-th best arm is pulled during these two phases. Based on Lemma \ref{lem:expl_pull}, these two terms can be bounded as:
	\begin{lemma}
		\label{lem:expl_decom}
		Conditioned on the typical event, the following two results hold simultaneously:
		\begin{equation}
		\small
		\notag
		\begin{aligned}
		\text{1) }&\text{For a sub-optimal arm $k$, }(\mu_{(M)}-\mu[k])T^{\text{expl}}_{k}(T) \\
		&\leq 20\min\left\{\frac{8\sqrt{6}\log(T)}{\mu_{(M)}-\mu[k]},\sqrt{T\log(T)}\right\};
		\\
		\text{2) }&\sum_{k\leq M}(\mu_{(k)}-\mu_{(M)})(T^{\text{expl}}-T^{\text{expl}}_{(k)}) \\
		&\leq 93\left(\sum_{k>M}\min\left\{\frac{8\sqrt{6}\log(T)}{\mu_{(M)}-\mu_{(k)}},\sqrt{T\log(T)}\right\}\right).
		\end{aligned}
		\end{equation}
	\end{lemma}
	
Finally, Lemma \ref{lem:expl_regret} can be proved by combining Eqn. \eqref{eqn:expl_decom} and Lemma \ref{lem:expl_decom}. 
\end{proof}

\subsection{Communication Regret}
The key result for the communication phases is that their regret does not dominate the overall regret, thanks to the expanded length of each exploration phase and the {adaptive quantization length} of arm statistics, as stated in Lemma \ref{lem:comm_regret}.

\begin{lemma}\label{lem:comm_regret}
	Conditioned on the typical event, the regret of the communication phase is bounded by:
	\begin{equation*}
 	\small
	\begin{aligned}
	&R^{\text{comm}}\leq 2\log_2 \left(\min\left\{\frac{8\sqrt{6}}{\Delta},\sqrt{T}\right\} \right)M^2KN'\left(L_H\right){\Delta_c}\\
	&+M^2\left(4\log_2 \left(\min\left\{\frac{8\sqrt{6}}{\Delta},\sqrt{T}\right\} \right)+K\right)N'(\lceil\log_2(K)\rceil){\Delta_c},
	\end{aligned}
\end{equation*}
	where $L_H=  1+\left\lceil\log_2\left(\min\left\{\frac{8\sqrt{3}}{\Delta},\sqrt{T}\right\}\right)\right\rceil$.
\end{lemma}
Note that $\log_2 \left(\min\left\{\frac{8\sqrt{6}}{\Delta},\sqrt{T}\right\} \right)$ becomes constant when $T$ is sufficient large. With {$N'(L)=O(\log(T))$}, asymptotically, the communication regret has an order of $O(\log(T))$, which is the same as the regret of other phases.

\subsection{Putting Everything Together}
We are now in the position to analyze the overall regret of EC3. With probability $P_s$, the typical event happens and the regret $R_s$ is the sum of the regret of $R^{\text{init}}, R^{\text{comm}}$ and $R^{\text{expl}}$ as analyzed in Lemma \ref{lem:init_regret}, \ref{lem:expl_regret} and \ref{lem:comm_regret}, respectively. With probability $1-P_s$, the atypical event happens and the regret can be linearly upper bounded as $R_f=MT{\Delta_c}$. Finally, the overall regret is the average of these two terms as
\begin{equation*}
\begin{aligned}
R(T)&=P_s R_s+(1-P_s)R_f\\
&\leq R^{\text{init}}+R^{\text{expl}}+R^{\text{comm}}+6M^2K{\Delta_c}\log_2(T).
\end{aligned}
\end{equation*}
Plugging in Lemmas \ref{lem:init_regret}, \ref{lem:expl_regret} and \ref{lem:comm_regret}, Theorem \ref{thm:overall_nosens} is proven.

\section{Extensions and Enhancements}
\label{sec:modelext}

\subsection{Collision-sensing model}
\label{subsec:sensing}

Although the proposed EC3 algorithm is primarily for no-sensing MP-MAB with collision-dependent rewards, it can be trivially extended to handle the collision-sensing counterpart. For the sake of completeness, we hereby highlight the necessary changes to EC3 for it to work in the collision-sensing setting, and present the corresponding regret analysis.
\subsubsection{Algorithm design}
In the collision-sensing model, the same EC3 design can be applied with the exception that error-correction coding is no longer necessary. {Adaptively} quantized sample means can be directly transmitted via forced collisions. Algorithms  \ref{alg:overall}, \ref{alg:init}, \ref{alg:comm_leader}, and \ref{alg:comm_follower} can be applied to the collision-sensing setting by removing the functions \texttt{Encoder()} and \texttt{Decoder()}. 
\subsubsection{Regret analysis}
The following theorem for the regret of EC3 in the collision-sensing setting can be similarly derived by following the steps for Theorem \ref{thm:overall_nosens} but omitting the loss caused by coding in the initialization and communication phases.
\begin{theorem}[\textbf{Collision-sensing regret}]
	\label{thm:overall_sens}
	The regret of the collision-sensing EC3 algorithm is upper bounded as
		\begin{equation*}
			\small
			\begin{aligned}
				R_{\text{cs}}&(T)\leq 113\sum_{k>M}\frac{8\sqrt{6}\log(T)}{\mu_{(M)}-\mu_{(k)}} \\
				&+2M^2K\log_2 \left(\frac{8\sqrt{6}}{\Delta} \right){\Delta_c}+4M^2\log_2 \left(\frac{8\sqrt{6}}{\Delta} \right){\Delta_c}\\
				&+M^2K{\Delta_c}+2MK{\Delta_c}+2M^2K{\Delta_c}\log_2(T).
			\end{aligned}
		\end{equation*}
		Furthermore, the instance-independent regret of the collision-sensing EC3 algorithm is upper bounded as
		\begin{equation*}\small
			R_{\text{cs}}(T)\leq O\left(K\sqrt{T\log(T)}\right).
		\end{equation*}
\end{theorem}
We see that the asymptotic regret of the collision-sensing EC3 algorithm is $R_{\text{cs}}(T){\leq}O\left(\sum_{k> M}\frac{\log(T)}{\mu_{(M)}-\mu_{(k)}}\right)$, which is of the same order as the lower bound in Eqn.~\eqref{eqn:lower_bound}.

\subsection{Practical error-correction codes}\label{subsec:code}

In practice, several existing coding techniques, although not optimal in the error exponent sense, can achieve the error rate of $\frac{1}{T}$ with a coding length $N=\Theta(\log(T))$. This leads to {longer coding lengths (and thus higher communication cost) than using $N'(L)$}, but does not change the regret scaling. {For a message of length $L$}, detailed theoretical analysis with decoding error bounds for practical coding techniques reveals the following code length requirement: repetition code $N_{rep}=L\lceil\frac{8\sigma^2\log(2LT)}{(\mu_{\min}-\nu_{\max})^2}\rceil$, Hamming code $N_{ham}=\frac{7L}{4}\lceil\frac{4\sigma^2\log(6LT)}{(\mu_{\min}-\nu_{\max})^2} \rceil$ and convolutional code $N_{con}=3L\lceil\frac{16\sigma^2\log(2^7LT)}{7(\mu_{\min}-\nu_{\max})^2} \rceil$. See Appendix~\ref{app:code} for details. It is worth noting that the difference becomes more prominent when $T$ is large, which is verified in the experiments.

\subsection{Unknown time horizon}
In the EC3 algorithm described in Section \ref{sec:alg}, prior knowledge of the time horizon $T$ is required. We note that this assumption can be relaxed with the standard doubling trick \cite{auer2010ucb,besson2018doubling}, which leads to an anytime version of the EC3 algorithm that can maintain the regret of order $O(\log(T))$ without a known time horizon.

\subsection{Refined collision-dependent reward model}
\label{subsec:refine}

The model in Section~\ref{sec:prob} assumes that the reward distribution when collision happens (i.e., $B_k$ for arm $k\in[K]$) does not depend on how many players are involved in the collision. In reality, however, it is likely that more players involved in the collision leads to worse rewards for all of them. Hence, one can adopt a more refined reward distribution that depends on the number of colliding players $\gamma$ as follows: if there are $\gamma\geq 2$ players involved in the collision on arm $k$, the reward $X_{k,\gamma}^{B}$ is sampled from a distribution $B_{k,\gamma}$ with $\mathbb{E}[X_{k,\gamma}^{B}]=\nu_\gamma[k]$ and when $\gamma\leq \omega$, $\nu_\gamma[k]\geq \nu_{\omega}[k]$. We further make an assumption\footnote{{This assumption avoids the situation that the sum reward of two players colliding on the same arm is higher than these two players playing their optimal non-colliding arms. It is an interesting future research direction to design MP-MAB algorithms that do not depend on this assumption.}} that the optimal arm-player assignment always have at most one player on each of the top-$M$ arms without collision. In other words, the optimal choice has a mean reward of $\sum_{m\in[M]}\mu_{(m)}$.

The EC3 algorithm can be directly applied to this refined collision-dependent reward model, because the communication phase is designed to strictly have $2$ players in collision while all three other phases are independent of the number of colliding players. The regret analysis also directly applies because for typical event (successful initialization, communications, and estimations), no more than $2$ players will create collision while for the atypical event, we have enforced a worst-case (linear) regret, meaning that it applies to an arbitrary number of colliding players.

\subsection{{Extension of Assumption \ref{aspt:bounds}}}\label{subsec:ext_aspt1}

For the sake of a clear regret definition and a better illustration of the EC3 scheme and the proof, Assumption \ref{aspt:bounds} is enforced in Section \ref{sec:prob}. It assumes that the upper bound for mean rewards upon collisions of all arms is strictly lower than the lower bound for the mean rewards with no collisions at arms\footnote{{Such  assumption is reasonable in practice, as channels upon collisions are often much worse than the ones without collisions.}}. However, this assumption is not fundamental, we can relax it to better characterize the collision-dependent rewards.

For example, we can assume such upper and lower bounds for each individual arm $k$ as $\nu[k]\leq \nu_{\max}[k] < \mu_{\min}[k]\leq \mu[k]$. Then, different error exponents (depending on $\nu_{\min}[k]$ and $\mu_{\max}[k]$) can be used to characterize each arm as a channel, e.g., $E_{r,k}$ for arm $k$, which further leads to the adoption of variable coding lengths for communication on arms. {A more careful communication loss characterization can also be performed, i.e., counting communication steps on each arm individually and using the specific $\Delta_{c}[k] = \mu_{(1)}-\nu[k]$ instead of the general $\Delta_c$ as the communication loss per step on arm $k$.} Such modifications do not fundamentally affect the original EC3 design but complicate notations and analysis unnecessarily, thus we focus on the simplified version in Assumption \ref{aspt:bounds}.

\subsection{More players than arms}

In practical applications it is possible to have more players than arms ($M>K$). In such scenarios, collisions are inevitable if all players have to play, and the players are in need of finding the arms with better rewards upon collisions, which poses a novel bandit learning problem \cite{bande2019multi}. Here we present some ideas of how to adapt EC3 to this setting. 

The action space for the players are denoted as $\mathcal{Z}=\prod_{m=1}^M\mathcal{Z}_m$, where $\mathcal{Z}_m=[K]$ is the individual action space for player $m$. The goal of the players now is to find the optimal matching in $\mathcal{Z}$ that maximizes their cumulative rewards. 
Inspired by \cite{boursier2019practical}, we treat each matching in $\mathcal{Z}$ as a `super-arm' and apply the EC3 algorithm correspondingly. Specifically, in the exploration phase, players can cooperatively and sequentially hop among the active matchings in $\mathcal{Z}$ to collect the reward statistics of these super-arms. The order of the sequential hopping can be pre-determined based on players' indices. Then, in the communication phase, the statistics of super-arms are sent to the leader, who eliminates the sub-optimal ones and sends back the exploration sequences for the next phase. During communication, error-correction coding can be similarly adopted to ensure communication quality. Lastly, when there is only one active super-arm, the players can begin exploration without further communications. We note that the theoretical analysis can be carried out for this scenario with a similar procedure in the proof of EC3 in Section \ref{sec:theory}.

\subsection{Other extensions}
There are several other interesting extensions for EC3 and the collision-dependent reward setting, which may inspire further research activities. First, a heterogeneous version of the problem can be studied, similar to the heterogeneous MP-MAB problem in \cite{bistritz2018distributed} but with collision-dependent reward distributions. This problem assumes different reward distributions for different players on the same arm, and we believe that the method in \cite{boursier2019practical} can be adapted in the framework of EC3. 
{Second, other than the stochastic reward setting, other MAB variants are also worth exploring in the case of multiple players. For example, adversarial rewards are studied in both collision-sensing \cite{alatur2019multi} and no-sensing \cite{bubeck2019non, shi2020no} setting, which is an interesting future direction for the collision-dependent reward model. 
Third, it is interesting but also challenging to remove Assumption \ref{aspt:bounds} in general. Some attempts have been made in \cite{bubeck2020coordination,bubeck2020cooperative} but with a focus on the instance-independent regret. It is worth exploring how to communicate without such assumptions while still getting an $O\left(\frac{1}{\Delta}\log(T)\right)$-style regret that approaches the centralized setting. We believe tools from communication theory related to information transmission over unknown channels can be useful.}

\section{Experiments}\label{sec:exp}
To validate the EC3 design, we have carried out numerical experiments using both synthetic and real-world datasets {with a focus on the no-sensing setting}.  All  results are obtained by averaging over $100$ independent runs. 

\subsection{Synthetic Dataset}
\label{subsec:exp_syn}
\begin{figure}[thb]
	\vspace{-0.1in}
	\setlength{\abovecaptionskip}{-2pt}
	\centering
	\includegraphics[width=0.9 \linewidth]{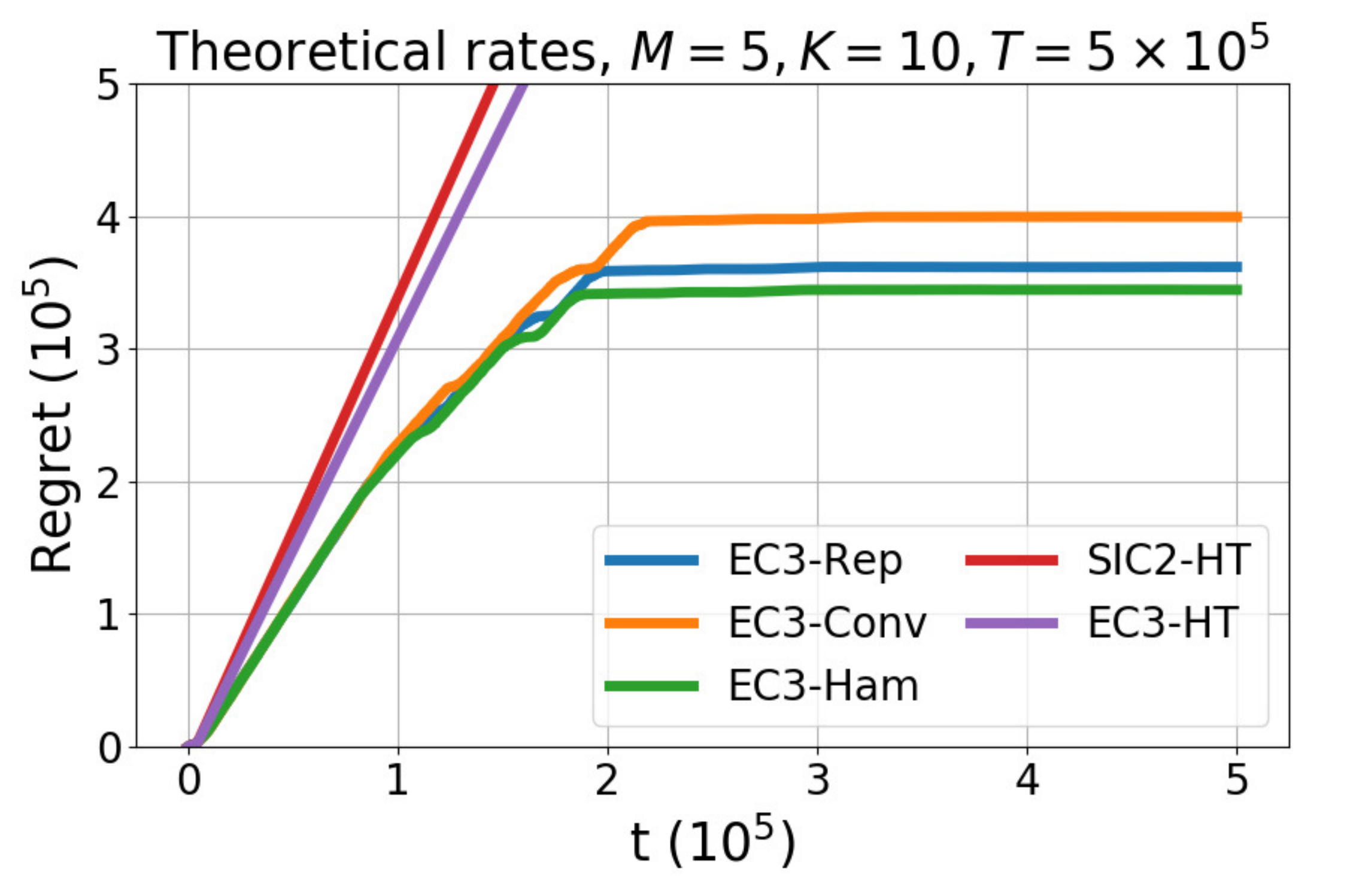}
	\caption{Regret comparison of different algorithms using a synthetic dataset and theoretically computed coding rates for EC3.}
	\label{fig:regret_compare}
	\vspace{-0.10in}
\end{figure}

Gaussian distributions $\sigma=0.2$ are used to generate a synthetic dataset and simulate the $0.2$-subgaussian reward distributions. {The bandit game is set to have $K=10$, and $M=5$ with no-collision mean rewards linearly distributed between $0.3$ and $0.84$ and collision mean rewards set as $0.1$ for all arms. Such an setting results in $\mu_{\min}=0.3, \nu_{\max} = 0.1$ and $\Delta = 0.06$. In each round of simulations, the mean rewards are uniformly interleaved among arms. } First, the performance of a modified SIC-MMAB2 \cite{boursier2018sic} is shown (labeled as `SIC2-HT'). {As SIC-MMAB2 is originally proposed for the no-sensing collision-eliminated reward setting, we enhanced it with hypothesis testing, where rewards less than the threshold $\frac{1}{2}(\mu_{\min}+\nu_{\max})$ are directly taken as collisions and otherwise no collisions}. {This enhancement enables SIC-MMAB2 to perform in the collision-dependent reward setting.} As shown in Fig. \ref{fig:regret_compare}, it cannot successfully converge due to a large amount of communication errors during communications. A similar phenomenon can be observed by adopting EC3 with hypothesis testing instead of coding (labeled as `EC3-HT'), which also has an almost linear regret. {These results illustrate the failure of using only hypothesis testing but not coding during communications with the collision-dependent rewards, which reinforces our intuition that the boundary between collision and non-collision rewards is very murky in this setting.}
	
Then, we perform the full version of EC3 with the coding, where three ``off-the-shelf'' error-correction codes: repetition code, Hamming code and convolutional code have been included. Details of the specific codes are provided in Appendix~\ref{app:code}. We first adopt the coding lengths computed theoretically in Section \ref{subsec:code} for these three different codes, which lead to the coding rates to be approximately: $R_{c,rep} \approx 6.25\times 10^{-3}$ for repetition code, $R_{c,ham}\approx 6.63\times 10^{-3}$ for Hamming code and $R_{c,conv}\approx 5.62\times 10^{-3}$ for convolutional code. As shown in Fig.~\ref{fig:regret_compare}, we see that the EC3 algorithm successfully converges with all three codes, which indicates the success of incorporating coding into communications and the effectiveness of the EC3 scheme.

{Although using the theoretically computed coding rates allows a guaranteed performance as in Fig. \ref{fig:regret_compare}, the inequalities used in the coding length computations are typically loose, which in practice result in unnecessarily small coding rates and costly communications. Thus, it is important to use more suitable coding rates in practice in order to control the communication cost, which is the main focus of the following experiments. We hence adopt a fixed coding rate of $R_c \approx 0.018$, which is much higher than the theoretically computed ones, to evaluate these coding schemes. As shown in Fig. \ref{fig:regret_compare_2}, EC3 still converges with Hamming code and has a lower regret compared with Fig.~\ref{fig:regret_compare}, which coincides with the intuition that the theoretically computed coding length is practically unnecessary. However, with repetition code and convolutional code, the regrets trend upwards, which indicate a certain amount of communications have failed with these two coding schemes under this ``extreme'' coding rate.}

\begin{figure}[thb]
	\vspace{-0.1in}
	\setlength{\abovecaptionskip}{-2pt}
	\centering
	\includegraphics[width=0.9 \linewidth]{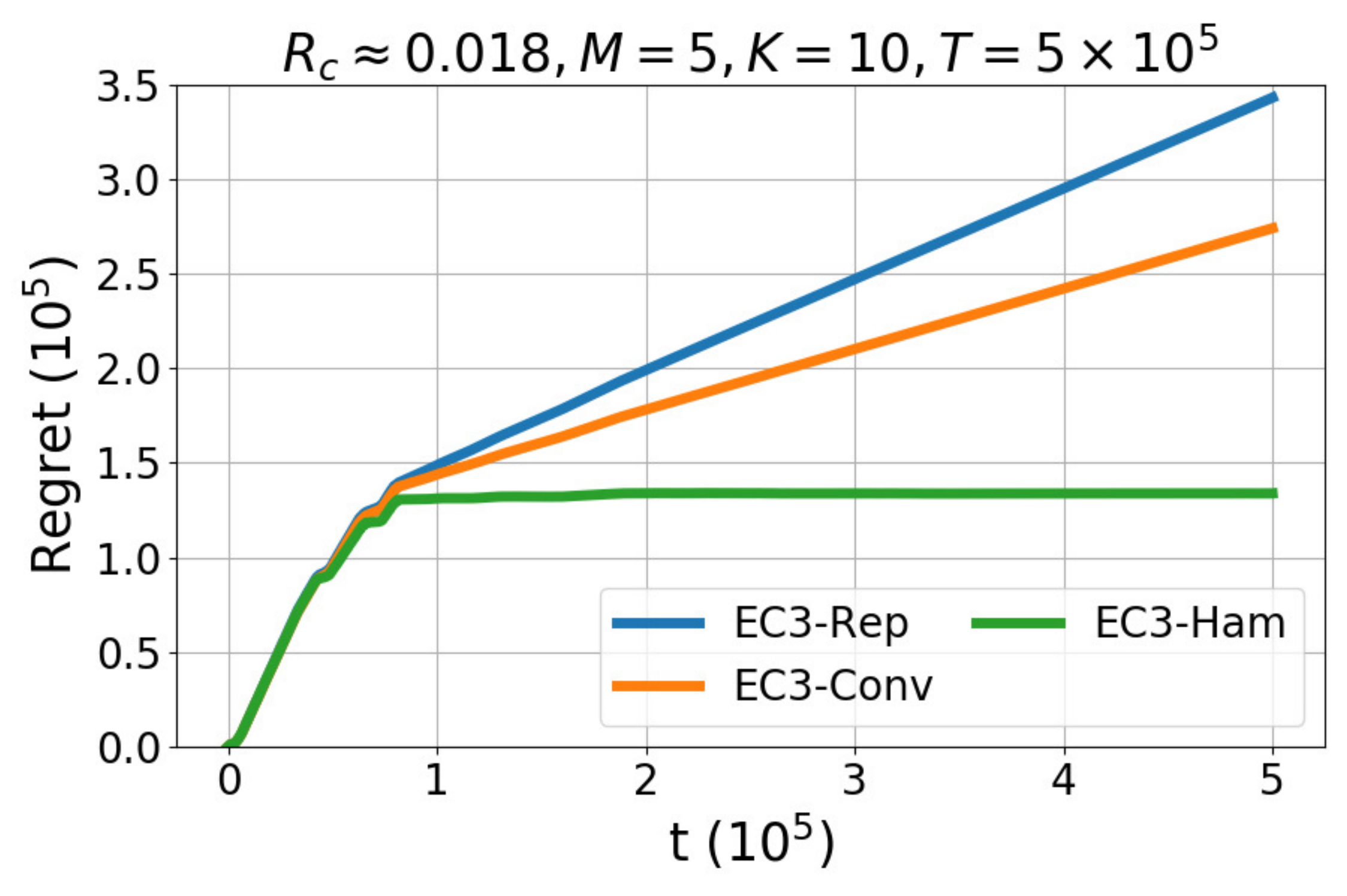}
	\caption{Regret comparison of {EC3 with different coding techniques} using a fixed high coding rate and a synthetic dataset.}
	\label{fig:regret_compare_2}
	\vspace{-0.10in}
\end{figure}

\begin{figure}[thb]
	\vspace{-0.1in}
	\setlength{\abovecaptionskip}{-2pt}
	\centering
	\includegraphics[width=0.9 \linewidth]{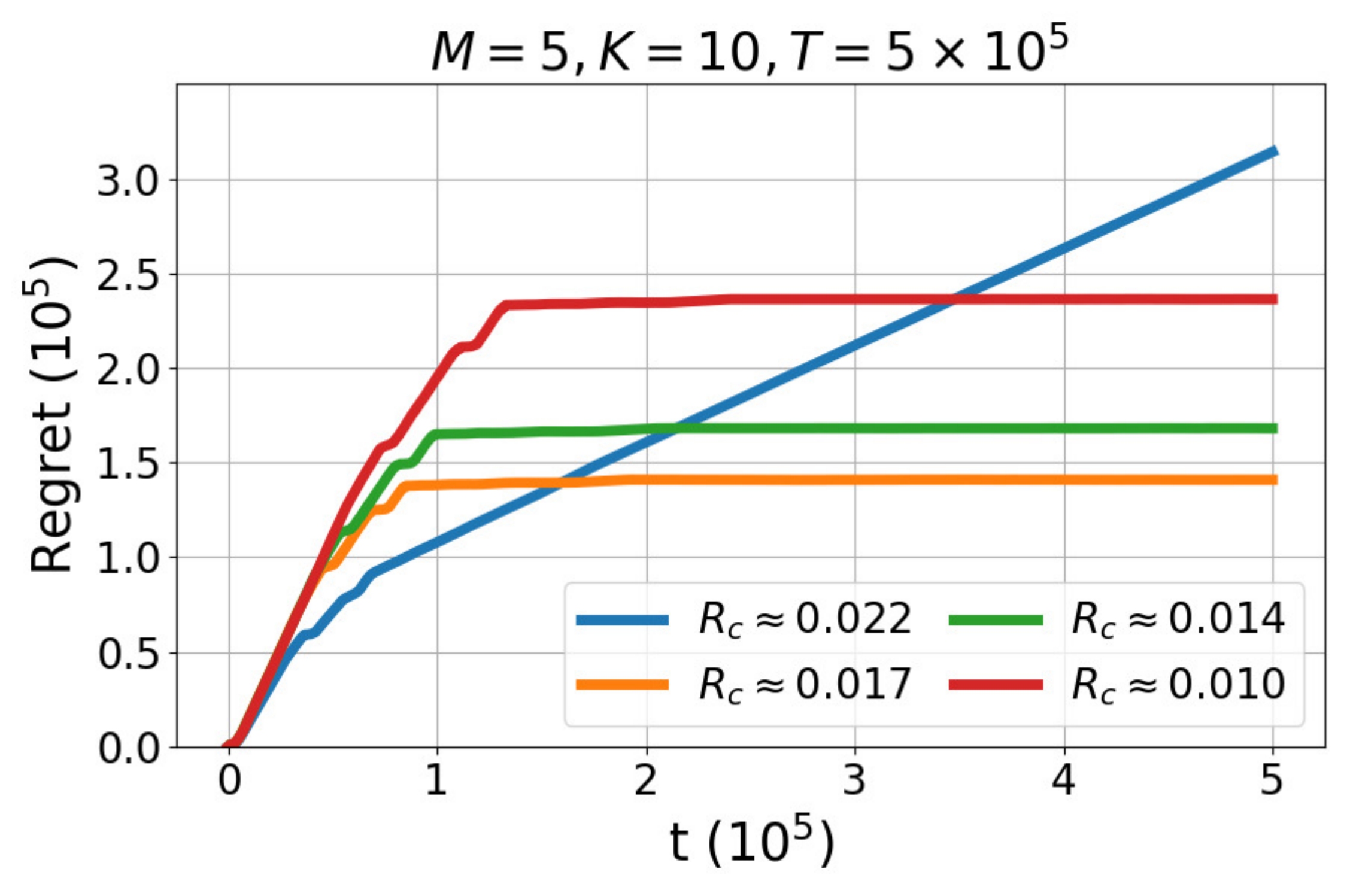}
	\caption{{Regret comparison of EC3 with {Hamming code} using different coding rates and a synthetic dataset.}}
	\label{fig:ham_rate}
	\vspace{-0.10in}
\end{figure}

Next we focus on Hamming code and test performance of EC3 using it with different coding rates. A very interesting tradeoff is observed from Fig.~\ref{fig:ham_rate}: when the coding rate decreases (e.g., from $0.022$ to $0.017$), the error-correction ability increases and there are fewer errors in the communication phases, which improves the performance. However, as we further decrease the rate (e.g., below $0.017$), the additional time slots needed to communicate the increased coded bits outweigh the benefit of improved error-correction ability, resulting in an increased overall regret. We further characterize this important tradeoff between communication regret and bit error rate in Fig.~\ref{fig:ham_error_regret}, where we see that although the bit error rate is monotonically increasing with the coding rate, the overall regret first decreases and then increases. This result highlights the importance of choosing a proper coding rate; our experiments show that $R_c=\frac{(\mu_{\min}-\nu_{\max})^2}{4\log(T)}$ usually balances error rate and communication regret well.

\begin{figure}[thb]
	\vspace{-0.1in}
	\setlength{\abovecaptionskip}{-2pt}
	\centering
	\includegraphics[width=0.9 \linewidth]{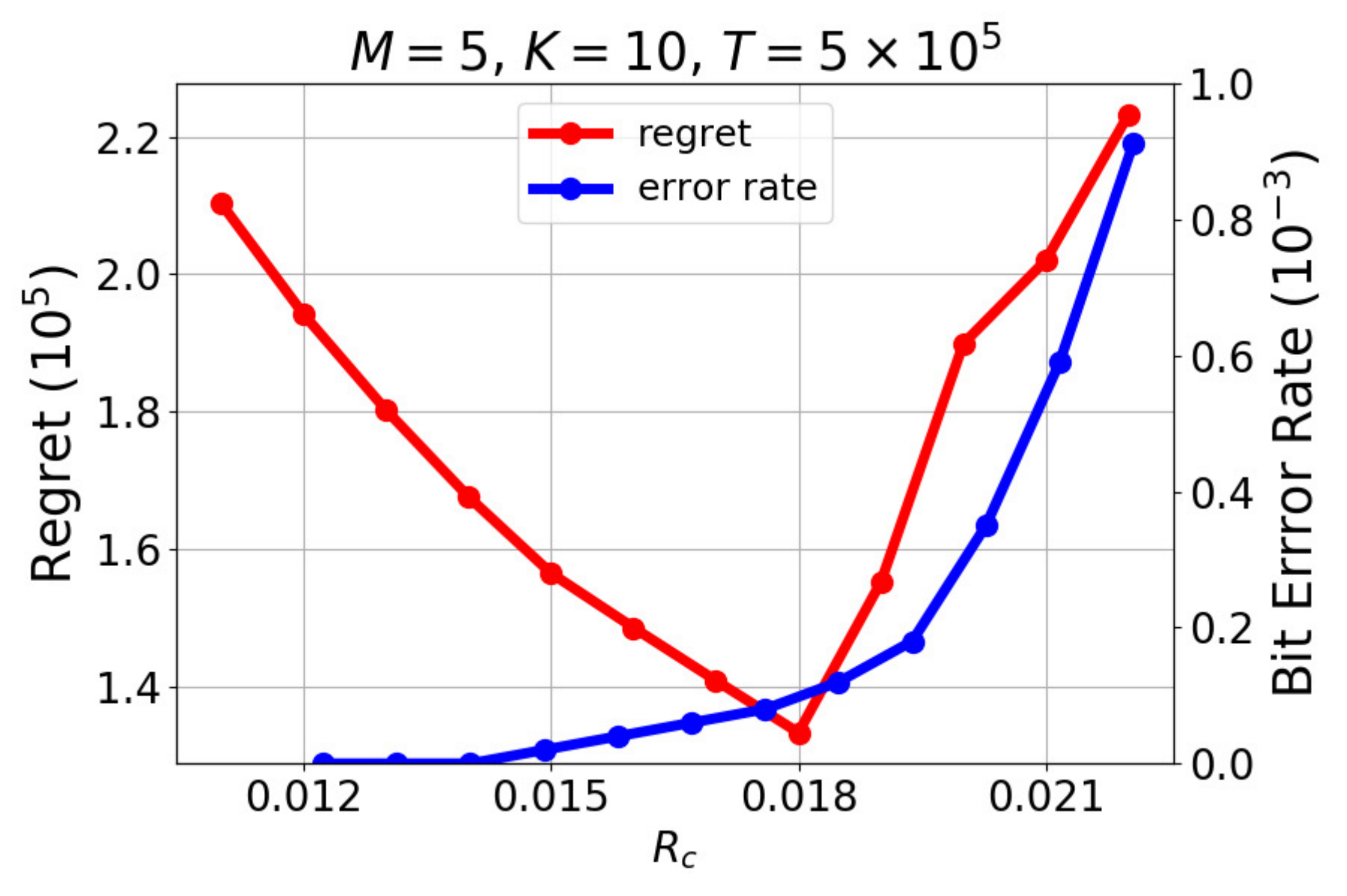}
	\caption{Regret and error probability {with Hamming code} as functions of coding rate.}
	\label{fig:ham_error_regret}
	\vspace{-0.10in}
\end{figure}

{Furthermore, the EC3 algorithm with Hamming code is performed with varying time horizons. As shown in Fig. \ref{fig:regret_horizon}, it is clear that under two different coding rates, the regrets all increase sublinearly with the time horizon, which corroborate the regret analysis.}

\begin{figure}[thb]
	\vspace{-0.1in}
	\setlength{\abovecaptionskip}{-2pt}
	\centering
	\includegraphics[width=0.9 \linewidth]{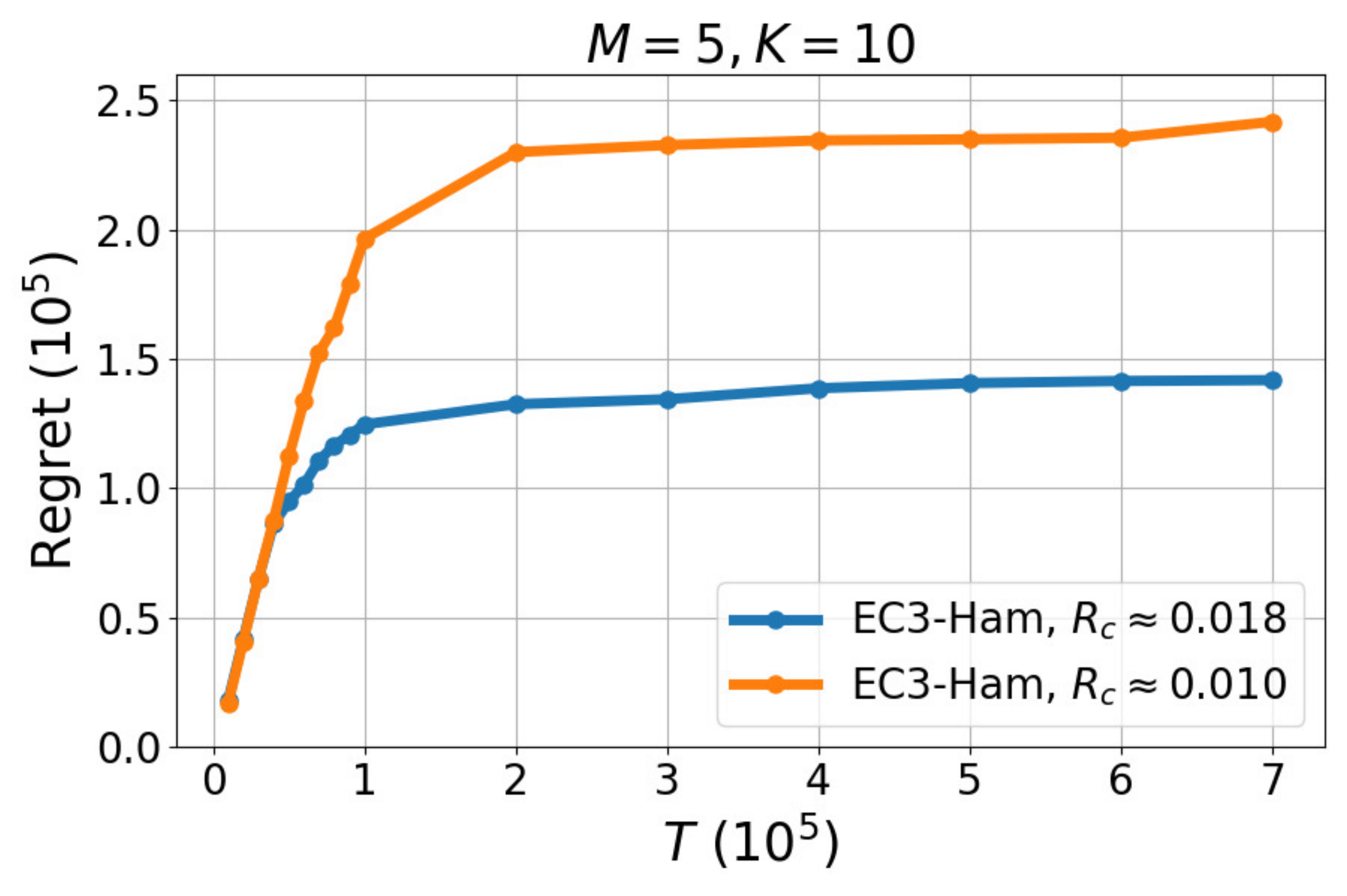}
	\caption{Regret with Hamming code as functions of time horizon.}
	\label{fig:regret_horizon}
	\vspace{-0.10in}
\end{figure}

{
Lastly, as the collision-eliminated reward setting is a special case of the general collision-dependent reward setting, i.e.,  $\nu_{\max} = 0$, we have also compared SIC-MMAB2 \cite{boursier2018sic} and EC3 in this setting. Fig. \ref{fig:regret_eliminate} shows that with the same no-collision mean rewards used in above experiments and zero rewards for collisions, although both algorithms converge successfully, EC3 with Hamming code with $R_c\approx 0.022$ outperforms SIC-MMAB2 significantly\footnote{{The coding rate $R_c\approx 0.022$ is feasible here but not in Fig.~\ref{fig:ham_rate} because in the collision-eliminated reward setting, the channel model is simpler as discussed in Section \ref{sec:info} where a lower error-correction capacity is sufficient.}}.
}
\begin{figure}[thb]
	\setlength{\abovecaptionskip}{-2pt}
	\centering
	\includegraphics[width=0.9 \linewidth]{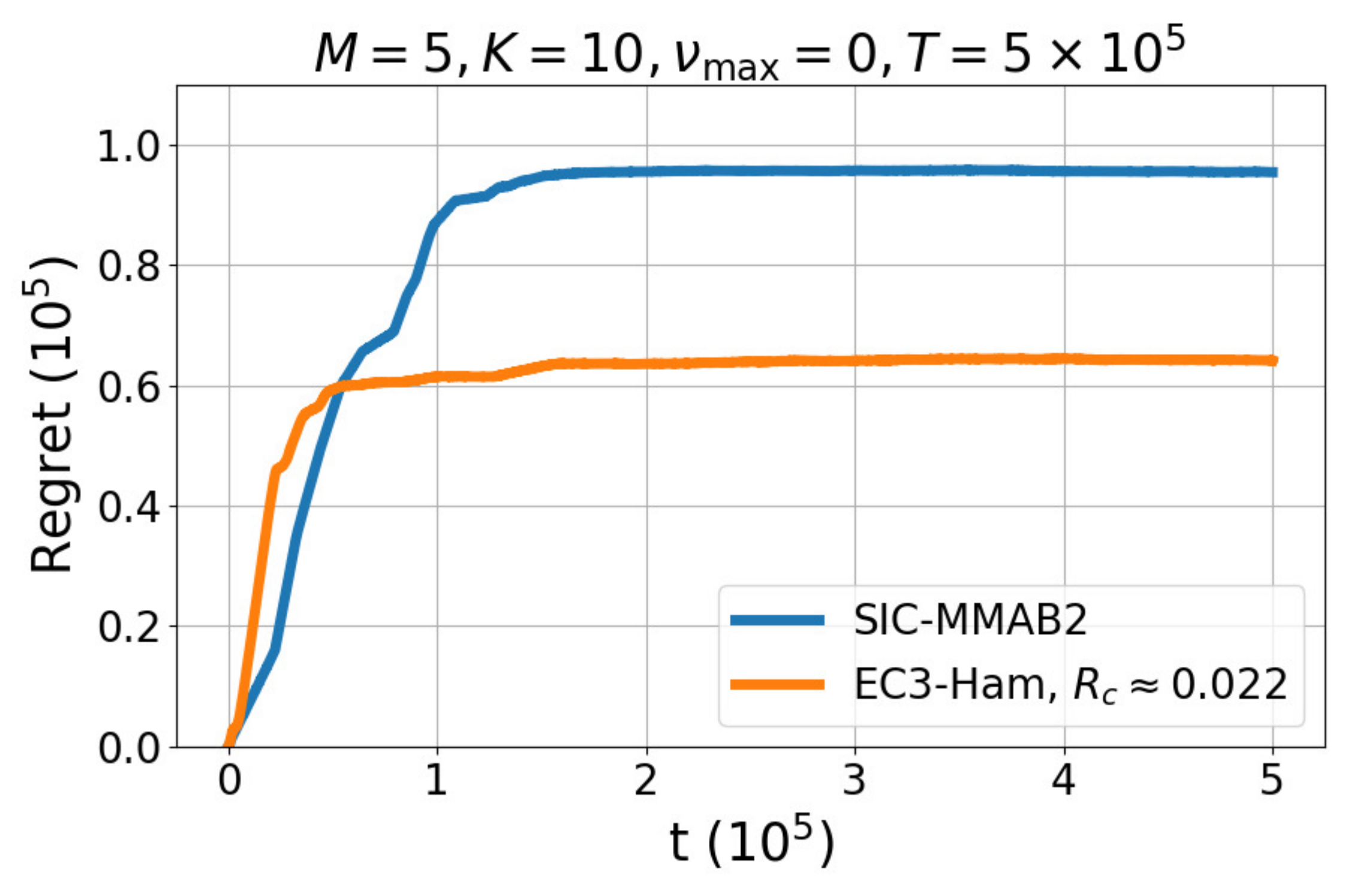}
	\caption{{Regret comparison of different algorithms in the collision-eliminated reward setting using a synthetic dataset.}}
	\label{fig:regret_eliminate}
	\vspace{-0.10in}
\end{figure}

\subsection{Real-world Datasets}

Due to the lack of applicable real-world datasets for cognitive radio systems, we adopt two real-world datasets from other scenarios: the NYC taxi dataset and Movielens dataset. Both represent relatively difficult bandit games that are useful in verifying the performance of EC3. These two datasets are used to mimic the underlying realization of channel quality of different channels in the application of cognitive radio. Since both datasets are not collected with MP-MAB and collision-dependent rewards in mind, we will discuss how we use the datasets for our purposes, and report the findings.

\begin{figure}[htb]
    \vspace{-0.15in}
    \setlength{\abovecaptionskip}{-2pt}
	\centering
	\subfigure[Low coding rate $R_c\approx 0. 0011$]{\setlength{\abovecaptionskip}{-2pt} \includegraphics[width=0.9\linewidth]{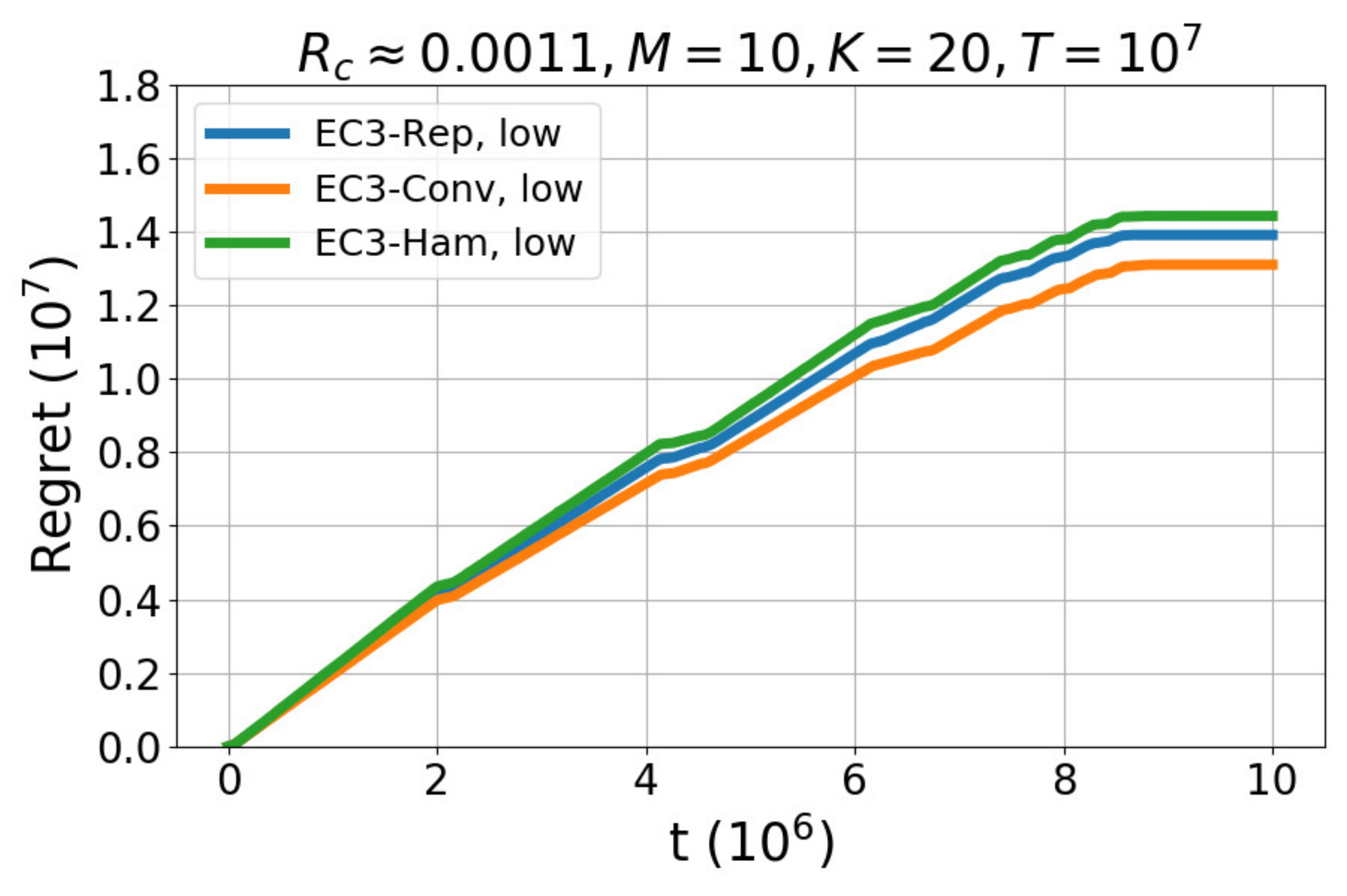}\label{fig:regret_ride_low}}
	\subfigure[High coding rate $R_c\approx 0. 0016$]{\setlength{\abovecaptionskip}{-2pt} \includegraphics[width=0.9\linewidth]{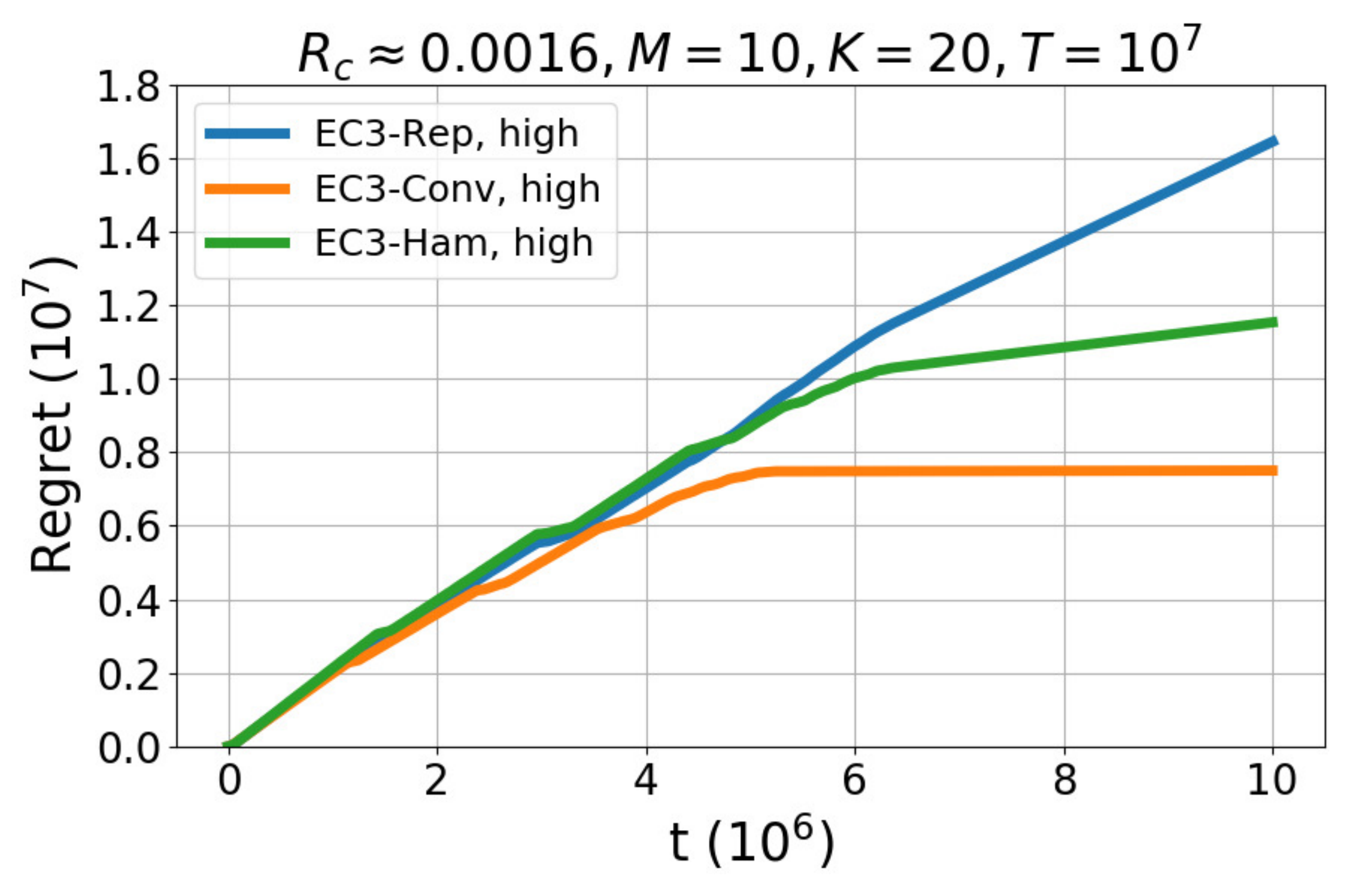}\label{fig:regret_ride_high}
	\vspace{-0.2in}}
	\caption{Regret comparison with different coding techniques using two coding rates and the NYC taxi dataset. }
	\label{fig:ride_16}
	\vspace{-0.1in}
\end{figure}

\textbf{NYC Taxi Dataset.} The taxi rides data from New York City in June, 2019 \cite{TLC2019DATA} consists of over 7 million taxi rides records of $265$ areas. The areas are randomly grouped into $40$ groups and assigned with rewards at time $t$ (minute) according to the number of taxi trips which took place in them. The rewards are then normalized to $[0,1]$, which forms distributions satisfying the $\frac{1}{2}$-subgaussian property. {The groups of areas can be interpreted as different channels in the cognitive radio setting, and the original taxi services as primary users' usage. More services in a certain group of areas are viewed as more intensive primary users' activities on a certain channel, which leads to lower rewards for the secondary users that we assumed to be engaged in the game.
}

The reward sequence is then replicated to create a final reward sequence of length $T=10^7$. The $20$ groups with larger reward means are used as the no-collision rewards, and the remaining $20$ groups are used as the collision rewards. $M=10$ players are assumed to engage in the game. The final reward sequence has $\mu_{\min}\approx 0.67$, $\nu_{\max}\approx 0.60$ and $\Delta\approx 0.03$. As shown in Fig.~\ref{fig:ride_16}, {with a low coding rate $R_c\approx 0.0011$, EC3 using all three codes successfully converges. With a higher coding rate $R_c\approx0.0016$, the optimal arms are perfectly learned by EC3 using convolutional code with a lower regret, while a slightly upward-trending regret for Hamming code indicates a small but non-negligible amount of communication errors have occurred, and the rapidly increasing regret of repetition code is due to its poor error-correction capability.}

\begin{figure}[htb]
\vspace{-0.1in}
	\setlength{\abovecaptionskip}{-2pt}
	\centering
	\subfigure[Low coding rate $R_c\approx 0. 0012$]{ \includegraphics[width=0.9\linewidth]{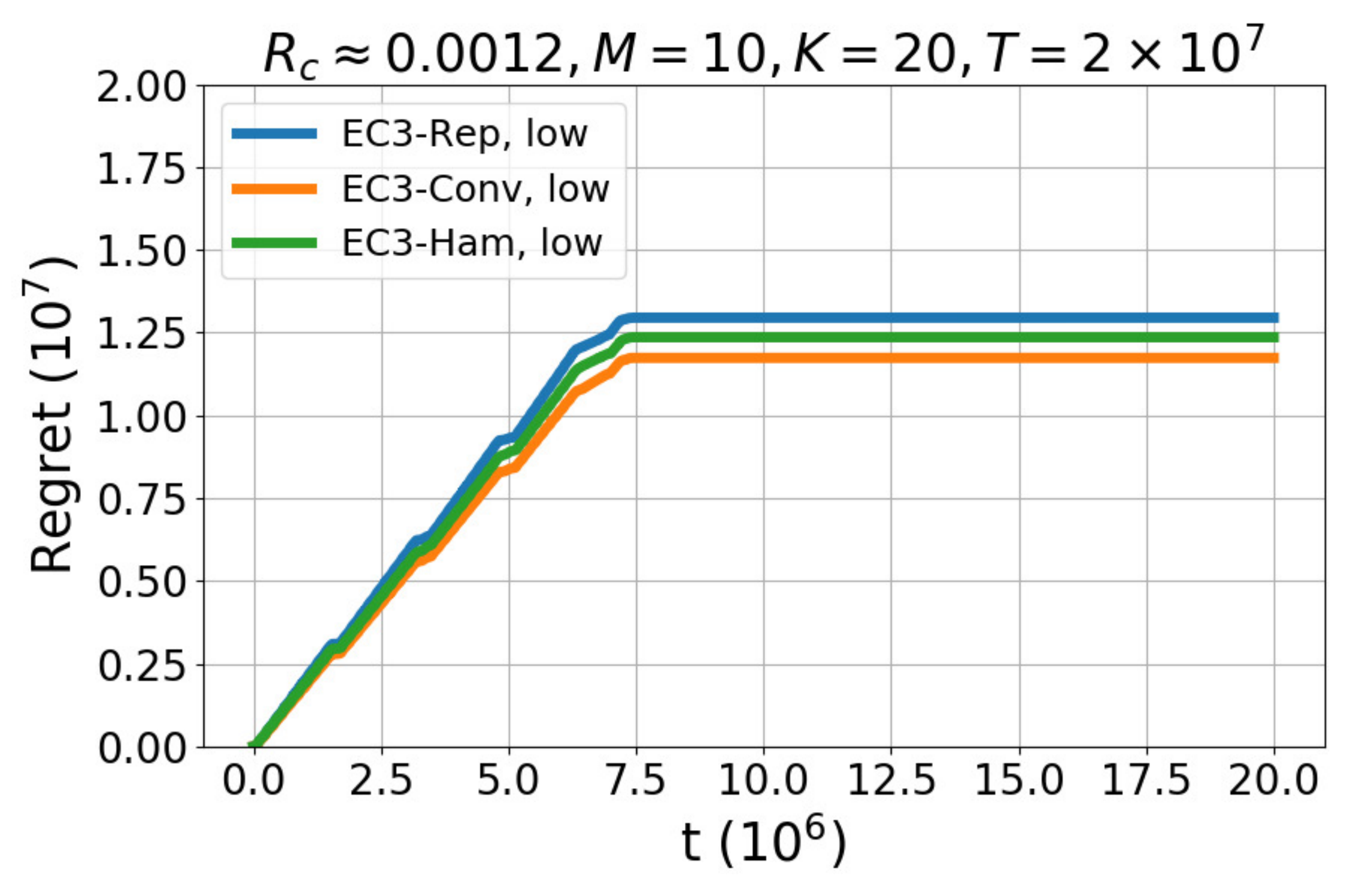}\label{fig:regret_movie_low}
	\vspace{-0.1in}}
	\subfigure[High coding rate $R_c\approx 0. 0017$]{ \includegraphics[width=0.9\linewidth]{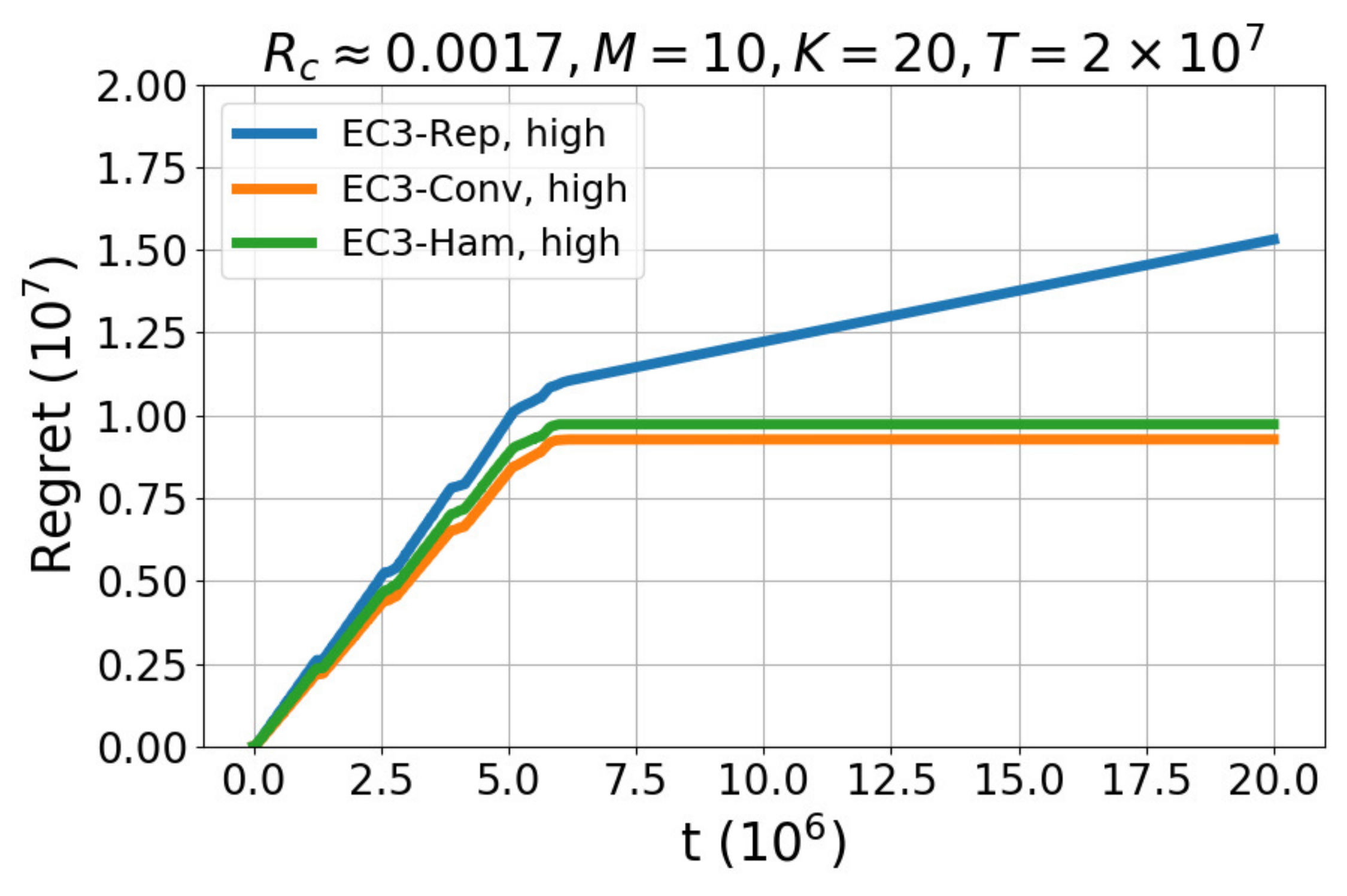}\label{fig:regret_movie_high}}
	\caption{Regret comparison with different coding techniques using two coding rates and the MovieLens dataset.}
	\label{fig:movie_17}
	\vspace{-0.1in}
\end{figure}
\textbf{MovieLens Dataset.} The EC3 algorithm is also evaluated on the popular MovieLens dataset \cite{harper2016movielens}. It consists of watching data of more than $2\times 10^4$ movies from over $10^5$ users between January 09, 1995 and March 31, 2015. The same grouping and reward-assigning procedures described above are used as pre-processing. {Similar interpretations can be applied as follows. The groups of movies can be interpreted as different channels in the cognitive radio setting. More original users watching a certain group of movies correspond to more intensive activities of primary users on this channel, which lowers the reward of this channel for secondary users.}

This reward sequence is replicated to a final length $T=2\times 10^7$. The no-collision and collision reward sequences are also represented by groups with different means of rewards. $M=10$ players are assumed to engage in the game. The final reward sequence has $\mu_{\min}\approx 0.87$, $\nu_{\max}\approx0.85$ and $\Delta\approx 0.04$. {As illustrated in Fig.~\ref{fig:movie_17}, with a low coding rate $R_c\approx 0.0011$, EC3 with all three codes successfully find the optimal arms. With a high coding rate $R_c\approx 0.0017$, convolutional code and hamming code still guarantee successful communications and lead to lower regrets, however, non-negligible communication errors occur with repetition code and the regret line trends upward.}  It can also be observed that with a larger $T$ in real-world datasets, convolutional code and Hamming code have better performance than repetition code, which corroborates our theoretical analysis.

\section{Conclusion}
\label{sec:conc}
In this work, we have introduced the decentralized MP-MAB problem with collision-dependent rewards, and proposed the EC3 algorithm to solve this new model {with a focus on the} no-sensing setting. An information-theoretic channel model was introduced for implicit communications, and random coding error exponent was utilized for the optimal communication in the no-sensing setting. With this careful communication design, quantized arm statistics can be shared among players at a level that is reliable enough to prevent the communication regret from dominating the overall regret. By expanding the exploration phases and {adaptively quantizing statistics for communication}, the theoretical analysis showed that EC3 can approach the lower bound of the centralized MP-MAB for the collision-dependent reward model.  Experiments with both synthetic and real-world datasets and several practical error-correction codes proved the superiority of the EC3 algorithm under different bandit configurations, and highlighted the efficiency of coding in the no-sensing setting.

\section*{Acknowledgement}
We thank Wei Xiong for assistance with the experiments. 


\appendices
\section{Proof of Corollary \ref{col:coding_length}}\label{app:coding_length}
{
\begin{proof}
    With Eqn. \eqref{eqn:error_exponent} in Theorem \ref{thm:error_exponent}, to get an error probability lower than $\frac{1}{T}$ in transmitting a message of a fixed length $L$ is equivalent to let the block length $N$ satisfy
    {\small\begin{align}\label{eqn:coding_req}
        \exp\left[-NE_r\left(\frac{L}{N}\right)\right]\leq \frac{1}{T}\Leftrightarrow NE_r\left(\frac{L}{N}\right)\geq \log(T).
    \end{align}}%
    The optimal coding length in the error exponent sense, i.e.,  $N'(L)$, can thus be expressed as $N'(L) =\arg\min_{N}\left\{NE_r\left(\frac{L}{N}\right)\geq \log(T)\right\}$.
    Since $E_r(R_c)$ is monotonically decreasing with the coding rate $R_c = \frac{L}{N}$ \cite{gallager1968information}, 
    it is easy to verify that there is a code with block length $N=\max\left\{\frac{L}{C-\varphi},\frac{\log(T)}{E_r(C-\varphi)}\right\}$ that satisfies the condition in Eqn.~\eqref{eqn:coding_req}, where $\varphi$ is an arbitrary constant in $(0,C)$. Thus, we can reach the conclusion that $N'(L)\leq \max\left\{\frac{L}{C-\varphi},\frac{\log(T)}{E_r(C-\varphi)}\right\}$.
\end{proof}
}

\section{Details on error-correction codes}
\label{app:code}
Three representative coding techniques are applied in the experiments. Their details and some preliminary analysis are provided in this appendix section.

\subsection{Repetition Code}

In the no-sensing setting, although a single sample reward from an arm can take any arbitrary value in the support of the corresponding distribution, the average sample mean concentrates to the true reward mean. With this observation, a carefully designed repetition code is applied as follows. 

\noindent\textbf{Coding scheme.} Each information bit is repeated into a bit string of length $N_0=\lceil\frac{8\sigma^2}{(\mu_{\min}-\nu_{\max})^2}\log(2LT)\rceil$ at the encoder. At the receiver, a sample mean of the receive values during the $N_0$ times steps, {denoted as $\hat{\mu}$}, is calculated. If the sample mean is larger than the threshold $\frac{1}{2}(\mu_{\min}+\nu_{\max})$, the decoder declares bit $1$; otherwise it declares bit $0$. 

\noindent\textbf{Analysis.} The {1-0} bit error probability can be bounded as:
{\small\begin{align*}
p_1&=P({\boldsymbol{b}[l]}=1,{\hat{\boldsymbol{b}}[l]}=0)\\
&=P\left(\hat{\mu}\leq \frac{1}{2}(\mu_{\min}+\nu_{\max})\right)\\
&= P\left(\hat{\mu}-\mu[k]\leq \frac{1}{2}(\mu_{\min}+\nu_{\max})-\mu[k]\right)\\
&\leq P\left(\hat{\mu}-\mu[k]\leq -\frac{\mu_{\min}-\nu_{\max}}{2}\right) \\
&\leq \exp\left(-\frac{N_0(\mu_{\min}-\nu_{\max})^2}{8\sigma^2}\right)\\
&\leq \frac{1}{2LT}.
\end{align*}}%
A similar bound holds for $p_0=P({\boldsymbol{b}[l]}=0,{\hat{\boldsymbol{b}}[l]}=1)$. For the entire sequence, a union bound of $L$ bits leads to an upper bound for the communication error rate of $1/LT$. Thus, the overall code word length for a message of length of $L$ is:
\begin{equation*}\small
	N_{rep}=L\left\lceil\frac{8\sigma^2}{(\mu_{\min}-\nu_{\max})^2}\log(2LT)\right\rceil.
\end{equation*}

\subsection{Hamming Code}
To further improve the error-correction ability, a modified $(7,4)$ Hamming code can be applied. 

\noindent\textbf{Coding scheme.} It is a concatenated code, with the standard $(7,4)$ Hamming code as the inner code and a repetition code as the outer code.
\begin{itemize}
	\item Encoding: The standard (7,4) Hamming encoding matrix $\boldsymbol{G}$ is first used to encode a 4-bit message to a 7-bit codeword. Then we repeat each bit of the 7-bit codeword $A$ times, leading to a $7A$-bit codeword;
	\item Decoding: First by using the repetition code decoding rule with $\frac{1}{2}(\mu_{\min}+\nu_{\max})$ as the threshold, $7A$-bit coded message is decoded into $7$ bits. These $7$ bits are then decoded with matrix $\boldsymbol{H}$ (the standard (7,4) Hamming decoding matrix). The final output is a decoded $4$-bit message.
	$$ \small \boldsymbol{G}=
	\left(
	\begin{matrix}
	1& 1& 0& 1\\
	1& 0& 1& 1\\
	1& 0& 0& 0\\
	0& 1& 1& 1\\
	0& 1& 0& 0\\
	0& 0& 1& 0\\
	0& 0& 0& 1
	\end{matrix}
	\right),\ 
	\boldsymbol{H}=
	\left(
	\begin{matrix}
	1& 0& 1& 0& 1& 0& 1\\
	0& 1& 1& 0& 0& 1& 1\\
	0& 0& 0& 1& 1& 1& 1
	\end{matrix}
	\right)
	$$
\end{itemize}

\noindent\textbf{Analysis.} The repetition code can reduce the bit error probabilities to:
{\small
\begin{align*}
&p_1\leq \exp\left(-\frac{A(\mu_{\min}-\nu_{\max})^2}{8\sigma^2}\right);\\
&p_0\leq \exp\left(-\frac{A(\mu_{\min}-\nu_{\max})^2}{8\sigma^2}\right).
\end{align*}}
It can thus be viewed as a binary asymmetric channel (BSC) with crossover probability  \begin{equation*}\small
    p=\exp\left(-\frac{A(\mu_{\min}-\nu_{\max})^2}{8\sigma^2}\right).
\end{equation*}
When $p$ is small, the block error probability of the Hamming Code in BSC is bounded as \cite{mackay2003information}:
\begin{align*}
	p_B\simeq 21p^2.
\end{align*}
The overall error probability can be bounded using a union bound as:
\begin{equation*}\small
	P_e\leq \frac{L}{4}p_B\simeq \frac{21L}{4}p^2=\frac{21L}{4}\exp\left(-\frac{A(\mu_{\min}-\nu_{\max})^2}{4\sigma^2}\right).
\end{equation*}
By choosing $A=\lceil\frac{4\sigma^2}{(\mu_{\min}-\nu_{\max})^2}\log(6LT) \rceil$, the error probability is smaller than $\frac{1}{LT}$. Thus, the total codeword length for a message of length of $L$ is:
\begin{equation*}\small
	N_{ham}=\frac{7L}{4}\left\lceil\frac{4\sigma^2}{(\mu_{\min}-\nu_{\max})^2}\log(6LT) \right\rceil.
\end{equation*}

\subsection{Convolutional Code}
Convolutional code is another powerful coding technique for reliable communications. Here we again use a modified convolutional code, where we choose a convolutional code encoder as an inner code and a repetition code as the outer code. Viterbi decoding algorithm can be implemented in the decoder. Every convolutional code encoder can be characterized with parameter $R_{\text{conv}}$, $B_{\text{free}}$ and $d_{\text{free}}$.

\noindent\textbf{Coding scheme.}  The inner convolutional code used for experiments in Section \ref{sec:exp} is generated by a $(3,1,2)$-feedforward encoder with $R_{\text{conv}}=\frac{1}{3}$, $B_{\text{free}}=1$ and $d_{\text{free}}=7$.

\noindent\textbf{Analysis.}  Similar to the analysis of Hamming Code, communication with the modified convolutional codes of length $A$ can be modeled as a BSC with crossover probabilities $p=\exp\left(-\frac{A(\mu_{\min}-\nu_{\max})^2}{8\sigma^2}\right)$. The bit error probability for this convolutional code is bounded as: 
\begin{align*}
	p_e&\simeq B_{\text{free}}2^{d_{\text{free}}}p^{d_{\text{free}}/2}\\
	&=B_{\text{free}}2^{d_{\text{free}}} \exp\left(-\frac{d_{\text{free}}A(\mu_{\min}-\nu_{\max})^2}{16\sigma^2}\right).
\end{align*}
The overall error probability can be bounded as:
\begin{equation*}\small
	P_e\leq Lp_e=B_{\text{free}}2^{d_{\text{free}}}L\exp\left(-\frac{d_{\text{free}}A(\mu_{\min}-\nu_{\max})^2}{16\sigma^2}\right).
\end{equation*}
With $A=\lceil\frac{16}{d_{\text{free}}(\mu_{\min}-\nu_{\max})^2}\log(2^{d_{\text{free}}}LT) \rceil$, the error probability is less than $\frac{1}{LT}$. The total code word length is then:
\begin{equation*}\small
	N_{con}=R_{\text{conv}}L\left\lceil\frac{16\sigma^2}{d_{\text{free}}(\mu_{\min}-\nu_{\max})^2}\log(B_{\text{free}}2^{d_{\text{free}}}LT) \right\rceil.
\end{equation*}

\section{Proof of Lemma \ref{lem:init_regret}}
\label{app:lem:init_regret}
\begin{proof}
	Each step in the first part of initialization can be viewed as receiving a bit $1$ from a potential player, which lasts $N'(1)$ steps. If the first part is successful, the second part consists of $M-1$ communication steps {of length $N'(\lceil\log_2(M)\rceil)$}. Denoting $\zeta_k=\mathds{1}\{\text{successful communication on arm $k$}\}$, the success probabilities of the both initialization parts, {i.e., $P_{i,1}$ and $P_{i,2}$,} are bounded as:
	{\small \begin{align*}
		P_{i,1}& = P(\forall k\in\{2,...,K\}, \zeta_k=1)\\
		&=1-P(\exists k\in\{2,...,K\},\zeta_k=0)\\
		&\geq 1-\sum_{k=1}^K P(\zeta_k=0)\\
		&=1-\frac{K}{T},
	\end{align*}}
	and
	{\small \begin{align*}
		P_{i,2} &=P(\forall m\in\{2,...,M\},\zeta_m=1)\\
		&=1-P(\exists m\in\{2,...,M\},\zeta_m=0)\\
		&\geq 1-\sum_{m=1}^M P(\zeta_m=0)\\
		&=1-\frac{M}{T}.
	\end{align*}}
	Thus, 
	$P_i=P_{i,1}P_{i,2}\geq 1-\frac{(M+K)}{T}$. {Each step leads to a loss of at most $\Delta_c$.} The initialization regret can be bounded as:
	\begin{align*}
    		R^{\text{init}}&\leq {\left(MKN'(1)+M^2N'(\lceil\log_2(M)\rceil)\right)}{\Delta_c} \\
    		&\leq 2MKN'(\lceil\log_2(K)\rceil){\Delta_c}.
	\end{align*}
	This completes the proof.
\end{proof}

\section{Proof of Lemma ~\ref{lem:success_com}}
\label{app:lem:success_com}
\begin{proof}
    {For time horizon $T$, there are at most $\log_2(T)$ communication phases, and in each phase, there are at most $K$ active arms. Since each follower only sends statistics to the leader and receives the length and the content of the set of accepted and rejected arms, there are at most $(K+2)M\log_2(T)+KM$ communications between the followers and the leader. The choice of coding length, i.e., $N'(L)$, guarantees that the decoding error happens with probability less than $\frac{1}{T}$ each time, thus by using a simple union bound, we have}
    \begin{align*}
        P_r&= 1-\mathbb{P}(\text{at least one message is decoded incorrectly})\\
        &\geq 1-\left[(K+2)M\log_2(T)+KM\right]\cdot \frac{1}{T}\\
        & \geq 1-\frac{3MK\log_2(T)}{T},
    \end{align*}
    which completes the proof.
\end{proof}
\section{Proof of Lemma \ref{lem:success_est}}
\label{app:lem:success_est}

\begin{proof}
    Lemma~\ref{lem:success_est} ensures event $A_3$ happens with a high probability, which further guarantees that the acceptance and rejection of arms are successful with a high probability. With {$1+Q_p$ bits to quantize sample means in phase $p$}, the probability that the sample mean exceeds the quantization range can be bounded. In the case of successful quantization, the estimation error consists of two parts: the quantization error and the sample uncertainty, which can be analyzed separately.

	With $1$ bit representing the integer part, the sample mean can be correctly quantized if the value does not exceed $2$. Thus, denoting $\chi_p^m[k]=\mathds{1}\{$successful quantization in round $p$ on arm $k$ by player $m\}$,  this probability can be bounded as:
	{\small
	\begin{align*}
	P(\chi_p^m[k]=0)&=P(\hat{\mu}_p^m[k]>2)\\
	&\leq P(\hat{\mu}_p^m[k]-\mu[k]>1)\\
	&\overset{(a)}{\leq} P\left(\hat{\mu}_p^m[k]-\mu[k]>\sqrt{\frac{2\log(T)}{T_p^m}}\right)\\
	&\overset{(b)}{\leq}  \frac{1}{T},
	\end{align*}}
	where inequality (a) is because $\max_p\left\{\sqrt{\frac{2\log(T)}{T_p^m}}\right\}\leq 1$ since $\min_p\{T_p^m\}=2\lceil\log(T)\rceil$ when $p=1$ and inequality (b) is from the subgaussian property.

	{With the choice of $Q_p = \left\lceil\log_2\left(\frac{1}{B_{T_p}}\right) \right\rceil \geq\log_2\left(\frac{1}{B_{T_p}}\right)$, the quantization error of arm $k$ in the case of $\chi_p^m[k]=1$, $\forall m\in[M]$ can be bounded as:}
	\begin{align*}
	&|\bar{\mu}_p^m[k]-\hat{\mu}_p^m[k]| \leq {2^{-Q_p}\leq B_{T_p},}
	\end{align*}
	and
    \begin{align*}
	&|\bar{\mu}_p[k]-\hat{\mu}_p[k]| = \frac{|\sum_{m=1}^M(\bar{\mu}_p^m[k]-\hat{\mu}_p^m[k])\cdot T_p^m|}{T_p}\leq {B_{T_p}}.
	\end{align*}
	Then, the overall gap between quantized mean and true mean can be analyzed from the two scenarios that the quantization is failed or successful. Denote $H$ as the overall rounds of explorations and communications {and $\chi[k]=\cap_{p=1}^H\cap_{m=1}^M \chi_p^m[k]$}, we first have:
	{\small\begin{align*}
		&P\left(\exists p \leq H,  |\bar{\mu}_p[k]-\mu[k]|>{2}B_{T_p}|\chi[k]=1\right)\\
		&\leq {\sum_{p=1}^H P\left( \left|\bar{\mu}_p[k]-\hat{\mu}_p[k]+\hat{\mu}_p[k]-\mu[k]\right| >{2}B_{T_p}\bigg|\chi[k]=1\right)}\\
		&\leq {\sum_{p=1}^H} P\left(|\bar{\mu}_p[k]-\hat{\mu}_p[k]|+|\hat{\mu}_p[k]-\mu[k]| >{2}B_{T_p}\bigg|\chi[k]=1\right)\\
		&\leq {\sum_{p=1}^H} P\left(|\hat{\mu}_p[k]-\mu[k]|> B_{T_p}\bigg|\chi[k]=1\right)\\
		&\leq \frac{2H}{T}.
	\end{align*}}
	Then, we can also get
	\small{
	\begin{align*}
		P(\chi[k]=0)\leq \sum_{p=1}^H \sum_{m=1}^{M_p}P\left(\chi_p^m[k]=0\right)\leq\frac{HM}{T}.
	\end{align*}}%
	Overall, we can conclude that
	\begin{align*}
	P&\left(\exists p \leq H, |\bar{\mu}_p[k]-\mu[k]|>{2}B_{T_p}\right)\\
	=&P\left(\exists p \leq H,  |\bar{\mu}_p[k]-\mu[k]|>{2}B_{T_p}|\chi[k]=1\right)\cdot P\left(\chi[k]=1\right)\\
	+&P\left(\exists p \leq H,  |\bar{\mu}_p[k]-\mu[k]|>{2}B_{T_p}|\chi[k]=0\right)\cdot P\left(\chi[k]=0\right)\\
	\leq& P\left(\exists p \leq H,  |\bar{\mu}_p[k]-\mu[k]|>{2}B_{T_p}|\chi[k]=1\right)+P\left(\chi[k]=0\right)\\
	\leq& \frac{(M+2)H}{T}.
	\end{align*}

	With at most $\log_2(T)$ rounds of communication and exploration phases and $K$ arms, the probability that event $A_3$ holds is:
	\begin{align*}
		P_c&\geq 1-\frac{(M+2)K\log_2(T)}{T}\geq 1-\frac{2MK\log_2(T)}{T},
	\end{align*}
	which completes the proof.
\end{proof}

\section{Proof of Lemma \ref{lem:expl_pull}}
\label{app:lem:expl_pull}

Lemma~\ref{lem:expl_pull} controls the number of times an arm is pulled before being accepted or rejected, and is essential to controlling the rounds of exploration and communication. The proof is similar to \cite[Proposition 1]{boursier2018sic}.
\begin{proof}
	For an optimal arm $k$, we denote $\Delta_k = \mu[k]-\mu_{(M+1)}$, and $s_k$ as the {smallest} integer such that ${8}B_{s_k}\leq \Delta_k$. Simple manipulation shows
	\begin{equation*}\small
	s_k=\frac{{128}\log(T)}{\Delta_k^2}\geq\frac{128\log(T)}{(\mu[k]-\mu_{(M+1)})^2}.
	\end{equation*}
	For some $p$ such that $T_{p-1}\leq s_k<T_p$, we have:
	\begin{align*}
	\Delta_{k} &\geq {8}B_{T_p}\\
	|\bar\mu_p[k]-\mu[k]|&\leq {2}B_{T_p}\\
	|\bar\mu_p[i]-\mu[i]|&\leq {2}B_{T_p},\text{ for all sub-optimal active arms } i,
	\end{align*}
	{where the second and third inequalities are based on the typical event.} With the above inequalities, we can get that for any suboptimal arm $i$, it holds that
	\begin{align*}
		\bar\mu_p[k]-2B_{T_p}&\geq \mu[k]-4B_{T_p}\\
		&\geq \mu[k]-\mu_{(M+1)}+\mu[i]-4B_{T_p}\\
		& = \Delta_{k}-8B_{T_p}+\mu[i]+4B_{T_p}\\
		&\geq \mu[i]+4B_{T_p}\\
		&\geq \bar\mu_p[i]+2B_{T_p},
	\end{align*}
	which means that at time $T_p$, arm $k$ is accepted. 
	
	Recall that the number of times an active arm has been pulled up to phase $p$ is $T_p = \sum_{q=1}^{p}M_q2^q\lceil \log(T)\rceil$. With a non-increasing $M_p$, it holds that $T_{p+1}\leq 3T_p$. Thus $T_p \leq 3s_k = \frac{384\log(T)}{(\mu[k]-\mu_{(M+1)})^2}$, and arm $k$ is accepted after at most $\frac{384\log(T)}{(\mu[k]-\mu_{(M+1)})^2}$ pulls. The proof for rejecting sub-optimal arms is similar with $\Delta_k = \mu_{(M)}-\mu[k]$.
\end{proof}

\section{Proof of Lemma \ref{lem:expl_decom}}
\label{app:lem:expl_decom}

\begin{proof}

	We first prove part 1) of Lemma~\ref{lem:expl_decom}. Conditioned on the typical event, for a sub-optimal $k$, from Lemma~\ref{lem:expl_pull}, $T_k^{\text{expl}}(T)\leq \min\left\{\frac{{384}\log(T)}{(\mu_{(M)}-\mu[k])^2},T\right\}$, so we have:
	{\small \begin{align*}
	(\mu_{(M)}-\mu[k])T^{\text{expl}}_{k}(T)&=\min\left\{\frac{384\log(T)}{\mu_{(M)}-\mu[k]},(\mu_{(M)}-\mu[k])T\right\}\\
	& \overset{(a)}{=}\min\left\{\frac{384\log(T)}{\delta},\delta T\right\}\\
	& \overset{(b)}{\leq}8\sqrt{6}\min\left\{\frac{8\sqrt{6}\log(T)}{\delta},\sqrt{T\log(T)}\right\},
	\end{align*}}
	where equality (a) denotes $\delta = \mu_{(M)}-\mu[k]$ and inequality (b) is because  $\min\left\{\frac{384\log(T)}{\delta},\delta T\right\}$ is maximized at $\delta = 8\sqrt{\frac{6\log(T)}{T}}$, which yields the first part.

To prove part 2) of Lemma~\ref{lem:expl_decom}, we first need to establish Lemma~\ref{appendix:lem_part1} and Lemma~\ref{appendix:lem_part2}.
\begin{lemma}
	\label{appendix:lem_part1}
	Define $\hat{t}_k$ as the number of exploratory pulls before accepting/rejecting arm $k$ and $H$ as the overall number of exploration and communication phases, conditioned on the typical event, it holds that	%
	{\small
	\begin{align*}
	\sum_{k\leq M}&(\mu_{(k)}-\mu_{(M)})\left(T^{\text{expl}}-T^{\text{expl}}_{(k)}\right)
	\leq\sum_{j>M}\sum_{k\leq M}\sum_{p=1}^H 2^p\\&\cdot\lceil \log(T) \rceil(\mu_{(k)}-\mu_{(M)})
	\mathds{1}\left\{\min\{\hat{t}_{(j)},\hat{t}_{(k)}\}\geq T_{p-1}\right\}.
	\end{align*}}
\end{lemma}
	To prove Lemma \ref{appendix:lem_part1} holds, we note that during phase $p$, if an optimal arm $k$ has already been accepted, it will be pulled $K_p2^p\lceil \log(T) \rceil$ times. If this arm is still active (i.e., $\hat{t}_k>T_{p-1}$), it will be pulled $M_p2^p\lceil \log(T) \rceil$ times, which means that it is not pulled $(K_p-M_p)2^p\lceil \log(T) \rceil$ times.
	Thus, it holds that $T^{\text{expl}}_k\geq T^{\text{expl}}-\sum_{p=1}^H 2^p(K_p-M_p)\lceil\log(T)\rceil\mathds{1}\{\hat{t}_k> T_{p-1}\}$. Noticing that $K_p-M_p = \sum_{j>M}\mathds{1}\{\hat{t}_{(j)}> T_{p-1}\}$, we have $T^{\text{expl}}_k\geq T^{\text{expl}}-\sum_{p=1}^H \sum_{j>M}2^p\lceil\log(T)\rceil \mathds{1}\{\min\{\hat{t}_{(j)},\hat{t}_{k}\}>T_{p-1}\}$, which proves the lemma. We note that this lemma converts the expression of the time that an optimal arm is not pulled into the time that a sub-optimal arm is pulled.

\begin{lemma}
	\label{appendix:lem_part2}
	Conditioned on the typical event, we have:
	{\small
	\begin{align}
		\label{eqn:applempart2}
	\sum_{k\leq M}\sum_{p=1}^H 2^p \lceil &\log(T) \rceil(\mu_{(k)}-\mu_{(M)})\mathds{1}\left\{\min\{\hat{t}_{(j)},\hat{t}_{(k)}\}\geq T_{p-1}\right\}\notag\\
	&\leq {93}\min \left \{\frac{\log(T)}{\mu_{(M)}-\mu_{(j)}},\sqrt{T\log(T)} \right\}.
	\end{align}}
\end{lemma}

	To prove Lemma~\ref{appendix:lem_part2}, let us define $A_j$ as the left hand side of inequality \eqref{eqn:applempart2}.
	From Lemma \ref{lem:expl_pull}, we have
	\begin{equation*}
	    \small
	    \hat{t}_{(k)} \leq \min\left\{\frac{384\log(T)}{(\mu_{(k)}-\mu_{(M+1)})^2},T\right\}.
	\end{equation*}
	Thus if we denote $\Delta(p)=\sqrt{\frac{384\log(T)}{T_{p-1}}}$, the inequity $\hat{t}_{(k)}>T_{p-1}$ implies $\mu_{(k)}-\mu_{(M+1)}<\Delta(p)$. By denoting $\kappa_j$ as the smallest integer such that $\hat{t}_{(j)} \leq T_{\kappa_j}$, it follows:
	{\small
	\begin{align*}
	A_j &\leq \sum_{k\leq M}\sum_{p=1}^{\kappa_j} 2^p \lceil \log(T) \rceil \Delta(p) \mathds{1}\left\{\hat{t}_{(k)}\geq T_{p-1}\right\}\\
	&\leq\sum_{p=1}^{\kappa_j}\Delta(p)2^p\lceil \log(T) \rceil\sum_{k\leq M}\mathds{1}\left\{\hat{t}_{(k)}\geq T_{p-1}\right\} \\
	&= \sum_{p=1}^{\kappa_j}\Delta(p)2^p\lceil \log(T) \rceil M_p\\
	&\leq \sum_{p=1}^{\kappa_j}\Delta(p)(T_p-T_{p-1})\\
	&=384\log(T) \sum_{p=1}^{\kappa_j}\left(\frac{\Delta(p)}{\Delta(p+1)}+1\right)\left(\frac{1}{\Delta(p+1)}-\frac{1}{\Delta(p)}\right)\\
	&\overset{(a)}{\leq} 384(1+\sqrt{3}) \log(T)\sum_{p=1}^{\kappa_j}\left(\frac{1}{\Delta(p+1)}-\frac{1}{\Delta(p)}\right)\\
	&\leq 384(1+\sqrt{3})\log(T)\frac{1}{\Delta(\kappa_j+1)}\\
	&\overset{(b)}{\leq} 24\sqrt{2}(1+\sqrt{3})\sqrt{\hat{t}_{(j)}\log(T)}\\
	&\overset{(c)}{\leq} 24\sqrt{2}(1+\sqrt{3})\min \left \{\frac{8\sqrt{6}\log(T)}{\mu_{(M)}-\mu_{(j)}},\sqrt{T\log(T)} \right\}\\
	&\leq 93\min \left \{\frac{8\sqrt{6}\log(T)}{\mu_{(M)}-\mu_{(j)}},\sqrt{T\log(T)} \right\}
	\end{align*}}
	where inequality (a) is from $\left(\frac{\Delta(p)}{\Delta(p+1)}+1\right)=1+\sqrt{\frac{T_p}{T_{p-1}}}\leq 1+\sqrt{3}$ since $T_{p+1}\leq 3T_p$. {Inequality (b) is supported by $\Delta(\kappa_j+1)\geq \sqrt{\frac{128\log(T)}{\hat{t}_{(j)}}}$ from the observation that $\hat{t}_{(j)}\geq T_{\kappa_j-1}$ by the definition of $\kappa_j$ and $T_{p+1}\leq 3T_{p}$. Inequality (c) is from Lemma \ref{lem:expl_pull} that $\hat{t}_{(j)} \leq \min\left\{\frac{384\log(T)}{(\mu_{(M)}-\mu_{(j)})^2}, T\right\}$.}
\end{proof}
{With Lemmas \ref{appendix:lem_part1} and \ref{appendix:lem_part2}, we can finally get Lemma \ref{lem:expl_decom}.}

\section{Proof of Lemma \ref{lem:comm_regret}}
\label{app:lem:comm_regret}

\begin{proof}
    We denote $H$ as the number of exploration phases. First, the communication loss of sending arm statistics for $p\leq H$ is at most $\sum_{p=1}^H M^2KN'(L_p)$ conditioned on the typical event. Also, sending the cardinality of the accept/reject arm sets incurs a loss of $2M^2N'(\lceil\log_2(K)\rceil)H$. Next, transmitting the content of accepted and rejected arm sets at most incurs a loss of $M^2KN'(\lceil\log_2(K)\rceil)$ because there are at most $K$ arms to be accepted or rejected. Putting them together, the total communication loss is at most $\sum_{p=1}^H M^2KN'(L_p)+2M^2N(\lceil\log_2(K)\rceil)+M^2KN'(\lceil\log_2(K)\rceil)$.
    
    Lemma~\ref{lem:expl_pull} establishes that $T_H = \sum_{p=1}^H M_p 2^p \lceil \log(T) \rceil \leq 3\max_k \{s_k\}\leq\min\left\{\frac{384\log(T)}{(\mu_{(M)}-\mu_{(M+1)})^2},T\right\}$. We thus have an upper bound on $H$ as
	{\small
	\begin{align*}
		H &\leq \log_2\left(\min\left\{\frac{384}{(\mu_{(M)}-\mu_{(M+1)})^2},T\right\}\right)\\
		&\leq 2\log_2 \left(\min\left\{\frac{8\sqrt{6}}{\Delta},\sqrt{T}\right\} \right).
	\end{align*}}%
	Based on this, we can bound $R^{\text{comm}}$ {with $\Delta_c$} as
	{\small
	\begin{align*}
	&R^{\text{comm}}\\
	&\leq \sum_{p=1}^H \left(M^2KN'(L_p)+(2M^2H+M^2K)N'(\lceil\log_2(K)\rceil)\right){\Delta_c}\\
	&\overset{(a)}{\leq}  M^2KHN'(L_H){\Delta_c}+(2M^2H+M^2K)N'(\lceil\log_2(K)\rceil){\Delta_c}\\
	&{\leq} 2\log_2 \left(\min\left\{\frac{8\sqrt{6}}{\Delta},\sqrt{T}\right\} \right)M^2KN'\left(L_H\right){\Delta_c}\\
	&+M^2\left(4\log_2 \left(\min\left\{\frac{8\sqrt{6}}{\Delta},\sqrt{T}\right\} \right)+K\right)N'(\lceil\log_2(K)\rceil){\Delta_c},
	\end{align*}}%
	{where inequality (a) is because $N'(L_p)$ is monotonically increasing with $p$. With $L_H= 1+\left\lceil\log_2\left(\sqrt{\frac{T_H}{2\log(T)}}\right)\right\rceil\leq  1+\left\lceil\log_2\left(\min\left\{\frac{8\sqrt{3}}{\Delta},\sqrt{T}\right\}\right)\right\rceil$, Lemma~\ref{lem:comm_regret} is proved.}
\end{proof}

\bibliographystyle{IEEEtran}
\bibliography{MMABrevise}

\begin{thebibliography}{10}
\providecommand{\url}[1]{#1}
\csname url@samestyle\endcsname
\providecommand{\newblock}{\relax}
\providecommand{\bibinfo}[2]{#2}
\providecommand{\BIBentrySTDinterwordspacing}{\spaceskip=0pt\relax}
\providecommand{\BIBentryALTinterwordstretchfactor}{4}
\providecommand{\BIBentryALTinterwordspacing}{\spaceskip=\fontdimen2\font plus
\BIBentryALTinterwordstretchfactor\fontdimen3\font minus
  \fontdimen4\font\relax}
\providecommand{\BIBforeignlanguage}[2]{{%
\expandafter\ifx\csname l@#1\endcsname\relax
\typeout{** WARNING: IEEEtran.bst: No hyphenation pattern has been}%
\typeout{** loaded for the language `#1'. Using the pattern for}%
\typeout{** the default language instead.}%
\else
\language=\csname l@#1\endcsname
\fi
#2}}
\providecommand{\BIBdecl}{\relax}
\BIBdecl

\bibitem{Bubeck:2012}
S.~Bubeck and N.~Cesa-Bianchi, ``Regret analysis of stochastic and
  nonstochastic multi-armed bandit problems,'' \emph{Foundations and Trends in
  Machine Learning}, vol.~5, no.~1, pp. 1--122, 2012.

\bibitem{Wang2018tsp}
Z.~{Wang}, R.~{Zhou}, and C.~{Shen}, ``Regional multi-armed bandits with
  partial informativeness,'' \emph{IEEE Trans. Signal Processing}, vol.~66,
  no.~21, pp. 5705--5717, Nov. 2018.

\bibitem{Shen2019}
C.~Shen, ``Universal best arm identification,'' \emph{IEEE Trans. Signal
  Processing}, vol.~67, no.~17, pp. 4464--4478, Sept. 2019.

\bibitem{Thompson1933}
W.~Thompson, ``On the likelihood that one unknown probability exceeds another
  in view of the evidence of two samples,'' \emph{Biometrika}, vol.~25, no.
  3-4, pp. 285--294, December 1933.

\bibitem{Gittins1979}
J.~C. Gittins and D.~M. Jones, ``A dynamic allocation index for the discounted
  multiarmed bandit problem,'' \emph{Biometrika}, vol.~66, no.~3, pp. 561--565,
  1979.

\bibitem{Auer:2002}
P.~Auer, N.~Cesa-Bianchi, and P.~Fischer, ``Finite-time analysis of the
  multiarmed bandit problem,'' \emph{Mach. Learn.}, vol.~47, no. 2-3, pp.
  235--256, May 2002.

\bibitem{gai2014distributed}
Y.~Gai and B.~Krishnamachari, ``Distributed stochastic online learning policies
  for opportunistic spectrum access,'' \emph{IEEE Trans. Signal Processing},
  vol.~62, no.~23, pp. 6184--6193, 2014.

\bibitem{tekin2012online}
C.~Tekin and M.~Liu, ``Online learning in decentralized multi-user spectrum
  access with synchronized explorations,'' in \emph{IEEE Military
  Communications Conference (MILCOM)}.\hskip 1em plus 0.5em minus 0.4em\relax
  IEEE, 2012, pp. 1--6.

\bibitem{rosenski2016multi}
J.~Rosenski, O.~Shamir, and L.~Szlak, ``Multi-player bandits -- a musical
  chairs approach,'' in \emph{International Conference on Machine Learning},
  2016, pp. 155--163.

\bibitem{besson2017multi}
L.~Besson and E.~Kaufmann, ``Multi-player bandits revisited,'' in
  \emph{Proceedings of Algorithmic Learning Theory}, Apr. 2018, pp. 56--92.

\bibitem{boursier2018sic}
E.~Boursier and V.~Perchet, ``{SIC-MMAB}: Synchronisation involves
  communication in multiplayer multi-armed bandits,'' in \emph{Advances in
  Neural Information Processing Systems}, 2019, pp. 2249--2257.

\bibitem{boursier2019practical}
A.~Mehrabian, E.~Boursier, E.~Kaufmann, and V.~Perchet, ``A practical algorithm
  for multiplayer bandits when arm means vary among players,'' in
  \emph{International Conference on Artificial Intelligence and
  Statistics}.\hskip 1em plus 0.5em minus 0.4em\relax PMLR, 2020, pp.
  1211--1221.

\bibitem{liu2010distributed}
K.~Liu and Q.~Zhao, ``Distributed learning in multi-armed bandit with multiple
  players,'' \emph{IEEE Trans. Signal Processing}, vol.~58, no.~11, pp.
  5667--5681, 2010.

\bibitem{anandkumar2011distributed}
A.~Anandkumar, N.~Michael, A.~K. Tang, and A.~Swami, ``Distributed algorithms
  for learning and cognitive medium access with logarithmic regret,''
  \emph{IEEE J. Select. Areas Commun.}, vol.~29, no.~4, pp. 731--745, 2011.

\bibitem{avner2014concurrent}
O.~Avner and S.~Mannor, ``Concurrent bandits and cognitive radio networks,'' in
  \emph{Joint European Conference on Machine Learning and Knowledge Discovery
  in Databases}.\hskip 1em plus 0.5em minus 0.4em\relax Springer, 2014, pp.
  66--81.

\bibitem{CSMA2012}
K.~Xu, M.~Gerla, and S.~Bae, ``How effective is the {IEEE 802.11 RTS/CTS}
  handshake in ad hoc networks,'' in \emph{IEEE Global Telecommunications
  Conference}, vol.~1, Nov. 2002, pp. 72--76.

\bibitem{SesiaLTE}
S.~Sesia, I.~Toufik, and M.~Baker, \emph{{LTE - The UMTS Long Term Evolution:
  From Theory to Practice}}.\hskip 1em plus 0.5em minus 0.4em\relax Wiley,
  2011.

\bibitem{bande2019multi}
M.~Bande and V.~V. Veeravalli, ``Multi-user multi-armed bandits for
  uncoordinated spectrum access,'' in \emph{IEEE International Conference on
  Computing, Networking and Communications}, 2019, pp. 653--657.

\bibitem{bonnefoi2017multi}
R.~Bonnefoi, L.~Besson, C.~Moy, E.~Kaufmann, and J.~Palicot, ``Multi-armed
  bandit learning in {IoT} networks: Learning helps even in non-stationary
  settings,'' in \emph{International Conference on Cognitive Radio Oriented
  Wireless Networks}.\hskip 1em plus 0.5em minus 0.4em\relax Springer, 2017,
  pp. 173--185.

\bibitem{avner2016multi}
O.~Avner and S.~Mannor, ``Multi-user lax communications: a multi-armed bandit
  approach,'' in \emph{The 35th Annual IEEE International Conference on
  Computer Communications}.\hskip 1em plus 0.5em minus 0.4em\relax IEEE, 2016,
  pp. 1--9.

\bibitem{darak2019multi}
S.~J. Darak and M.~K. Hanawal, ``Multi-player multi-armed bandits for stable
  allocation in heterogeneous ad-hoc networks,'' \emph{IEEE Journal on Selected
  Areas in Communications}, vol.~37, no.~10, pp. 2350--2363, 2019.

\bibitem{tibrewal2019distributed}
H.~Tibrewal, S.~Patchala, M.~K. Hanawal, and S.~J. Darak, ``Distributed
  learning and optimal assignment in multiplayer heterogeneous networks,'' in
  \emph{IEEE INFOCOM 2019-IEEE Conference on Computer Communications}.\hskip
  1em plus 0.5em minus 0.4em\relax IEEE, 2019, pp. 1693--1701.

\bibitem{garivier2011kl}
A.~Garivier and O.~Capp{\'e}, ``The {KL-UCB} algorithm for bounded stochastic
  bandits and beyond,'' in \emph{Proceedings of the 24th annual conference on
  learning theory}, 2011, pp. 359--376.

\bibitem{proutiere2019optimal}
P.-A. Wang, A.~Proutiere, K.~Ariu, Y.~Jedra, and A.~Russo, ``Optimal algorithms
  for multiplayer multi-armed bandits,'' in \emph{International Conference on
  Artificial Intelligence and Statistics}.\hskip 1em plus 0.5em minus
  0.4em\relax PMLR, 2020, pp. 4120--4129.

\bibitem{magesh2019multi}
A.~Magesh and V.~V. Veeravalli, ``Multi-user mabs with user dependent rewards
  for uncoordinated spectrum access,'' in \emph{2019 53rd Asilomar Conference
  on Signals, Systems, and Computers}.\hskip 1em plus 0.5em minus 0.4em\relax
  IEEE, 2019.

\bibitem{alatur2019multi}
P.~Alatur, K.~Y. Levy, and A.~Krause, ``Multi-player bandits: The adversarial
  case,'' \emph{Journal of Machine Learning Research}, vol.~21, p.~77, 2020.

\bibitem{Shi2020aistats}
C.~{Shi}, W.~{Xiong}, C.~{Shen}, and J.~{Yang}, ``Decentralized multi-player
  multi-armed bandits with no collision information,'' in \emph{Proceedings of
  the 23rd International Conference on Artificial Intelligence and Statistics
  (AISTATS)}, Palermo, Sicily, Italy, 2020.

\bibitem{bubeck2020cooperative}
S.~Bubeck, T.~Budzinski, and M.~Sellke, ``Cooperative and stochastic
  multi-player multi-armed bandit: Optimal regret with neither communication
  nor collisions,'' \emph{arXiv preprint arXiv:2011.03896}, 2020.

\bibitem{bubeck2020coordination}
S.~Bubeck and T.~Budzinski, ``Coordination without communication: optimal
  regret in two players multi-armed bandits,'' in \emph{Conference on Learning
  Theory}.\hskip 1em plus 0.5em minus 0.4em\relax PMLR, 2020, pp. 916--939.

\bibitem{landgren2016distributed}
P.~Landgren, V.~Srivastava, and N.~E. Leonard, ``On distributed cooperative
  decision-making in multiarmed bandits,'' in \emph{European Control Conference
  (ECC)}.\hskip 1em plus 0.5em minus 0.4em\relax IEEE, 2016, pp. 243--248.

\bibitem{landgren2018social}
------, ``Social imitation in cooperative multiarmed bandits: partition-based
  algorithms with strictly local information,'' in \emph{IEEE Conference on
  Decision and Control (CDC)}.\hskip 1em plus 0.5em minus 0.4em\relax IEEE,
  2018, pp. 5239--5244.

\bibitem{martinez2019decentralized}
D.~Mart{\'\i}nez-Rubio, V.~Kanade, and P.~Rebeschini, ``Decentralized
  cooperative stochastic bandits,'' in \emph{Advances in Neural Information
  Processing Systems}, 2019, pp. 4529--4540.

\bibitem{wang2019distributed}
Y.~Wang, J.~Hu, X.~Chen, and L.~Wang, ``Distributed bandit learning:
  Near-optimal regret with efficient communication,'' in \emph{International
  Conference on Learning Representations}, 2020.

\bibitem{lugosi2018multiplayer}
G.~Lugosi and A.~Mehrabian, ``Multiplayer bandits without observing collision
  information,'' \emph{arXiv preprint arXiv:1808.08416}, 2018.

\bibitem{Lai:1985}
T.~L. Lai and H.~Robbins, ``Asymptotically efficient adaptive allocation
  rules,'' \emph{Adv. Appl. Math.}, vol.~6, no.~1, pp. 4--22, March 1985.

\bibitem{CoverBook}
T.~M. Cover and J.~A. Thomas, \emph{Elements of Information Theory},
  2nd~ed.\hskip 1em plus 0.5em minus 0.4em\relax USA: John Wiley \& Sons, Inc.,
  2006.

\bibitem{gallager1968information}
R.~G. Gallager, \emph{Information Theory and Reliable Communication}.\hskip 1em
  plus 0.5em minus 0.4em\relax Springer, 1968, vol.~2.

\bibitem{anantharam1987asymptotically}
V.~Anantharam, P.~Varaiya, and J.~Walrand, ``Asymptotically efficient
  allocation rules for the multiarmed bandit problem with multiple plays--part
  {I}: {IID} rewards,'' \emph{IEEE Trans. Autom. Control}, vol.~32, no.~11, pp.
  968--976, 1987.

\bibitem{auer2010ucb}
P.~Auer and R.~Ortner, ``{UCB} revisited: Improved regret bounds for the
  stochastic multi-armed bandit problem,'' \emph{Periodica Mathematica
  Hungarica}, vol.~61, no. 1-2, pp. 55--65, 2010.

\bibitem{besson2018doubling}
L.~Besson and E.~Kaufmann, ``What doubling tricks can and can't do for
  multi-armed bandits,'' \emph{arXiv preprint arXiv:1803.06971}, 2018.

\bibitem{bistritz2018distributed}
I.~Bistritz and A.~Leshem, ``Distributed multi-player bandits-a game of thrones
  approach,'' in \emph{Advances in Neural Information Processing Systems},
  2018, pp. 7222--7232.

\bibitem{bubeck2019non}
S.~Bubeck, Y.~Li, Y.~Peres, and M.~Sellke, ``Non-stochastic multi-player
  multi-armed bandits: Optimal rate with collision information, sublinear
  without,'' in \emph{Conference on Learning Theory}.\hskip 1em plus 0.5em
  minus 0.4em\relax PMLR, 2020, pp. 961--987.

\bibitem{shi2020no}
C.~Shi and C.~Shen, ``On no-sensing adversarial multi-player multi-armed
  bandits with collision communications,'' \emph{IEEE Journal on Selected Areas
  in Information Theory}, vol.~2, no.~2, pp. 515--533, 2021.

\bibitem{TLC2019DATA}
{NYC TLC}, ``{TLC} trip record data,''
  \url{https://www.nyc.gov/site/tlc/about/tlc-trip-record-data.page}, 2019.

\bibitem{harper2016movielens}
F.~M. Harper and J.~A. Konstan, ``The {MovieLens Datasets}: History and
  context,'' \emph{ACM Trans. Interact. Intell. Syst.}, vol.~5, no.~4, December
  2015.

\bibitem{mackay2003information}
D.~MacKay, \emph{Information Theory, Inference and Learning Algorithms}.\hskip
  1em plus 0.5em minus 0.4em\relax USA: Cambridge University Press, 2003.

\end{thebibliography}

\end{document}